\documentclass[a4paper,UKenglish]{lipics-v2018}

\theoremstyle{plain}
\newtheorem{proposition}[theorem]{Proposition}

\usepackage{amsmath,amssymb}
\usepackage{tikz}
\usetikzlibrary[arrows]

\newcommand{\emptyword}{\lambda}

\newcommand{\norm}[1]{\left\lVert #1 \right\rVert}

\newcommand{\N}{\mathbb{N}}

\newcommand{\cA}{\mathcal A}

\newcommand{\cV}{\mathcal V}

\newcommand{\cM}{\mathcal M}

\newcommand{\cF}{\mathcal F}
\newcommand{\cW}{\mathcal W}
\newcommand{\cO}{\mathcal O}

\newcommand{\bV}{\mathbf V}

\newcommand{\bW}{\mathbf W}

\newcommand{\Call}[1][A]{#1_C} 
\newcommand{\Return}[1][A]{#1_R} 
\newcommand{\Intern}[1][A]{#1_I} 

\newcommand{\Ext}{\textsc{Ext}}
 
\newcommand{\VPSigma}[1][A]{#1} 
\newcommand{\WMSigma}[1][A]{{#1}^\triangle} 
\newcommand{\FreeExt}[1][A]{\WMSigma} 
\newcommand{\ext}[1][a,b]{\mathrm{ext}_{#1}} 
\newcommand{\ExtHat}[1][A]{\widehat{\FreeExt[#1]}} 
\newcommand{\extomega}[1][ab]{\ext[#1]^\omega} 
\newcommand{\ExtMon}{\textbf{M}\Ext}
\newcommand{\LExtMon}{\cM \Ext}
\newcommand{\ExtMonV}[1][V]{\textbf{V}\Ext}
\newcommand{\LExtMonV}[1][V]{\cV \Ext}

\newcommand{\what}[1]{\widehat{#1}}
\newcommand{\mone}{{-1}}
\newcommand{\Eqat}[1]{[\![ #1 ]\!]}

\newcommand{\Ahat}[1][A]{\what{#1^*}}

\usepackage{microtype}

\title{Visibly Pushdown Languages and Free Profinite Algebras}

\titlerunning{Visibly Pushdown Languages and Free Profinite Algebras}

\author{Silke Czarnetzki}{Wilhelm-Schickard-Institut Universit\"{a}t T\"{u}bingen, Germany}{czarnetz@informatik.uni-tuebingen.de}{}{}
\author{Andreas Krebs}{Wilhelm-Schickard-Institut Universit\"{a}t T\"{u}bingen, Germany}{mail@krebs-net.de}{}{}
\author{Klaus-J\"{o}rn Lange}{Wilhelm-Schickard-Institut Universit\"{a}t T\"{u}bingen, Germany}{lange@informatik.uni-tuebingen.de}{}{}

\authorrunning{S. Czarnetzki, A. Krebs and K-J. Lange}

\Copyright{Silke Czarnetzki, Andreas Krebs and Klaus-Jörn Lange}

\subjclass{\ccsdesc[100]{Theory of computation~Algebraic language theory}}

\keywords{Visibly Pushdown Languages, Finite Algebra and Recognition, Pro-finite Algebra, Equations}

\category{}

\relatedversion{}

\supplement{}

\funding{}

\EventEditors{John Q. Open and Joan R. Access}
\EventNoEds{2}
\EventLongTitle{42nd Conference on Very Important Topics (CVIT 2016)}
\EventShortTitle{CVIT 2016}
\EventAcronym{CVIT}
\EventYear{2016}
\EventDate{December 24--27, 2016}
\EventLocation{Little Whinging, United Kingdom}
\EventLogo{}
\SeriesVolume{42}
\ArticleNo{23}
\nolinenumbers 

\begin{document}

\maketitle

\begin{abstract}
We build a notion of algebraic recognition for visibly pushdown languages by finite algebraic objects. These come with a typical Eilenberg relationship, now between classes of visibly pushdown languages and classes of finite algebras. Building on that algebraic foundation, we further construct a topological object with one purpose being the possibility to derive a notion of equations, through which it is possible to prove that some given visibly pushdown language is not part of a certain class (or to even show decidability of the membership-problem of the class in some cases). In particular, we obtain a special instance of Reiterman's theorem for pseudo-varieties. These findings are then employed on two subclasses of the visibly pushdown languages, for which we derive concrete sets of equations. For some showcase languages, these equations are utilised to prove non-membership to the previously described classes.
 \end{abstract}

\newpage

\section{Introduction}
The algebraic theory of finite monoids has led to many fruitful results on the regular languages. In particular, a lot of naturally emerging subclasses of regular languages are definable via the algebraic properties of their syntactic monoids \cite{Sc65,BrSi73} (for an overview, see the book of Straubing \cite{St94}). The deep understanding of this relationship between classes of regular languages and algebraic properties of finite monoids is mostly ascribable to Eilenberg's famous variety theorem \cite{Ei76}, but also Reiterman's equivalence between varieties and profinite equations \cite{Re82}, offers a different and enlightening perspective on these connections (see for instance the survey of Pin \cite{Pi12}). In addition, the description of varieties by a finite set of profinite equations, in many cases implicitly, as utilised in \cite{GPK16}, provides an algorithm for decidability of the membership problem of the variety.

In two papers \cite{GGP08,GGP10} Gehrke, Grigorieff and Pin point out that not only is Reiterman's theorem a consequence of the duality between Boolean algebras and topological spaces uncovered by Stone in 1936 \cite{St36}, but also that this duality can be utilised to obtain a notion of recognition and minimal recognising objects for Boolean algebras of arbitrary languages at one cost: the loss of some algebraic information. They prove that this topological object is a monoid if and only if the Boolean algebra contains only regular languages. 

In the light of this result, we examine a class of languages that is still relatively close to the regular languages in the sense that it forms a Boolean algebra and equivalence is decidable for the underlying machine model:
The visibly pushdown languages (VPL). \cite{AlMa04}

We prove that by diverging from the classical notion of monoids as algebraic recognisers and replacing them with certain finite algebras, it is possible to obtain a notion of finite algebraic recognisers for VPL. From those recognisers, we construct a profinite algebra, which like the recognisers proposed by Gehrke, Grigorieff and Pin is a Stone space. This particular space also preserves algebraic information on the structure of VPL.

Apart from a profinite object, we also state Eilenberg and Reiterman-like theorems that in combination allow for a characterisation of subclasses of VPL through algebraic properties and profinite equations. We then proceed to examine two subclasses of VPL: the visibly counter languages (VCL) and, as a subclass of VCL, the visibly counter languages with threshold zero. The decidability, whether a given VPL belongs to one of said classes was already shown in \cite{BLS06}.

We derive for both of these classes sound sets of equations and show on some examples, that in these cases, the equations provide a method to prove that a language is not a VCL or not a VCL with threshold zero.

\textbf{Organisation of the Paper}

Since the proofs of some theorems are a bit lengthy, each section is followed by a section containing the proofs. 

\section{Preliminaries}

Let $M$ be a monoid, then we denote by $1_M$ the neutral element of $M$. The neutral element of $A^*$ -- the empty word -- is denoted by $\lambda$. Moreover, we consider the set $M^M$ of all maps from $M$ to $M$ as a monoid with the concatenation of maps $\circ$.

\subsection{Visibly Pushdown Languages}

A \emph{visibly pushdown alphabet} is a finite alphabet $\VPSigma$, which is partitioned into three sets $\Call,\Return$ and $\Intern$. Letters in $\Call$ are named \emph{call letters}, while letters in $\Return$ are \emph{return letters} and letters in $\Intern$ \emph{internal letters}.

The visibly pushdown languages (VPL) were introduced in \cite{AlMa04} as the languages accepted by so-called \emph{visibly pushdown automata} (VPA), which are a restriction of pushdown automata in the sense that once a call letter is read, a symbol is pushed to the stack and similarly, return letters remove a symbol from the stack, while internal letters leave the stack untouched. We consider a slight restriction in the sense that the words accepted by the automata (and thus the VPL) are part of a subset of $A^*$ -- the so-called well-matched words -- which were already regarded by Alur et al. \cite{AKMV05}. 

\begin{definition}[Well-matched Words]
Let $A$ be a visibly pushdown alphabet, then the empty word $\lambda$ and each internal letter $c \in \Intern$ is well-matched, and inductively
\begin{itemize}
\item for a well-matched word $w$, $a \in \Call$ and $b \in \Return$, the word $awb$ is well-matched and
\item if $u,v$ are well-matched, then so is $uv$.
\end{itemize}
We denote the set of well-matched words over $A$ by $\WMSigma$.
\end{definition}

\subsection{Visibly Counter Languages}

Just like visibly pushdown automata are restricted pushdown automata, visibly counter automata (VCA) are restricted counter automata and also work over visibly pushdown alphabets accepting well-matched words. For VCA, reading a call letter increases the value of the counter by one, reading a return letter decreases the value by one and internal letters leave the value of the counter untouched. Special instances of visibly counter automata are the VCA with threshold $m \in \N$ (or $m$-VCA) and were treated, for instance, in \cite{BLS06}. Here, for any counter value smaller than $m$, the counter value may influence the behaviour of the automaton. For values greater than $m$, the automaton behaves essentially like a finite automaton, ignoring its counter value. As such, $0$-VCA are basically finite automata with an auxiliary device to ensure that the word is well-matched.
\section{VPL in Terms of Algebra}\label{sec:Algebra}

VPL were already characterised in \cite{AKMV05} through finite congruences on monoids. We enrich that result for VPL of well-matched words with additional algebraic structure: In fact, the set of all well-matched words $\WMSigma$ is a submonoid of $A^*$ that additionally supports a unary operation from $\WMSigma$ to $\WMSigma$, sending the word $w$ to $awb$ for $a$ a call and $b$ a return letter. One may visualise this operation as an \emph{extension} in height (by one) of the height profile of a well-matched word, which gives the algebraic objects defined in the following their name.

\subsection{Algebras and Morphisms}

We introduce the algebraic objects that form the foundation for the both algebraic and topological investigation of VPL. These algebraic objects will either be finite or free.

\begin{definition}[$\Ext$-Algebra]\label{def:exAlg}
An $\Ext$-algebra is a monoid $(R, \cdot)$ and a submonoid of $R^R$ denoted by $(\cO(R),\circ)$, which for each $r \in R$ contains the maps $x \mapsto r \cdot x$ and $x \mapsto x \cdot r$. 
We usually omit to mention $\cO(R)$ and say that $R$ is an $\Ext$-algebra.
\end{definition}

\textbf{Note on forest algebras.} Each $\Ext$-algebra $(R,\cO(R))$ is a forest algebra as introduced in \cite{BoWa08}, where the horizontal monoid is $R$, the vertical monoid $\cO(R)$ and the action of $\cO(R)$ on $R$ is function application. We distinguish them, since we are investigating languages of words rather than languages of trees. Still, it should be mentioned that VPL and regular tree languages have very close connections \cite{AlMa04}.
\\

Observe that the set of all well-matched words $\WMSigma$ is an $\Ext$-algebra: For any two words $u,v \in A^*$ such that $uv \in \WMSigma$ and $x \in \WMSigma$, let $\ext[u,v](x) = uxv$. Then we let $\cO(\WMSigma)$ be the set of all maps $\ext[u,v]$. The left- and right multiplication maps ($x \mapsto x \cdot r$ and $x \mapsto r \cdot x$) are given by $\ext[w,\lambda]$ (resp. $\ext[\lambda,w]$) for $w \in \WMSigma$.

\begin{definition}[Morphism]
Let $R$ and $S$ be $\Ext$-algebras. A morphism from $R$ to $S$ is a tuple of monoid morphisms $(\phi,\psi)$ with $\phi \colon R \rightarrow S$ and $\psi \colon \cO(R) \rightarrow \cO(S)$ such that for all $e \in \cO(R)$ and $r \in R$: $\psi(e)(\phi(r)) = \phi(e(r)).$
\end{definition}

Observe that $\phi$ is implicitly determined by $\psi$, since $\phi$ is monoid morphism and hence $\phi(1_R) = 1_S$. For $r,x \in R$ letting $m_r(x) = r \cdot x$, we obtain 
\[\phi(r) = \phi(m_r(1_R)) = \psi(m_r)(1_S).\]
Hence we cease to distinguish between $\phi$ and $\psi$ and say that $\psi \colon R \rightarrow S$ is a morphism of $\Ext$-algebras.

In particular, by the inductive definition of the well-matched words, any morphism $\psi \colon \WMSigma \rightarrow R$ is uniquely determined by its values on $\ext$ for $a \in \Call$ and $b \in \Return$, $\ext[c,\lambda]$ and $\ext[\lambda,c]$ for $c \in \Intern$.

In the following, we often write $\ext$ for the operation on $\WMSigma$ and also $\ext$ for an operation on some $\Ext$-algebra $R$. This often has its origin in the fact that the $\ext$ operation on $R$ is considered the morphic image of $\ext$ on $\WMSigma$ for some particular morphism $\psi \colon \WMSigma \rightarrow R$. We assume that it is understood from the context, which is which.

\subsection{Language Recognition}

Similar to recognition of regular languages by monoid morphisms, we can recognise languages of well-matched words via $\Ext$-algebra morphisms. While the syntactic monoid of a VPL, such as $\{a^nb^n \mid n \in \N\}$ is in general infinite, our notion of recognition through algebras instead of monoids and in particular the additional algebraic structure of $\Ext$-algebras allows us to obtain finite recognising objects for the non-regular VCL. This leads to the main theorem at the end of the section, stating that VPL are precisely the languages recognised by finite $\Ext$-algebras. 

\begin{definition}[Recognition]
A language $L \subseteq \FreeExt$ is recognised by an $\Ext$-algebra $R$, if there exists a morphism $\psi \colon \FreeExt \rightarrow R$, such that $L = \psi^{-1}(\psi(L))$.
\end{definition}

\begin{example}\label{ex:anbn}
Consider the $\Ext$-algebra $R$, where the left table displays the multiplication on $R$ and the right displays the maps with their respective values in $\cO(R)$:
\begin{equation*}\label{ex:extAlgebra}
\begin{tabular}[t]{c|ccc}
$\cdot$ & 1 & $x$ & $0$\\ \hline
$1$ & 1 & $x$ & $0$\\ 
$x$ & $x$ & $0$ & $0$\\ 
$0$ & $0$ & $0$ & $0$  
\end{tabular}
\hspace{1cm}
\begin{tabular}[t]{c|ccc}
& 1 & $x$ & 0\\ \hline
$\text{id}$ & $1$ & $x$ & 0\\
$\ext$ & $x$ & $x$ & 0\\
$\ext[x,1]$ & $x$ & $0$ & 0\\
$\ext[x,x]$ & $0$ & $0$ & 0
\end{tabular}
\end{equation*}
To be more precise, we consider the visibly pushdown alphabet with one call letter $a$, one return letter $b$ and no neutral letters. Then the morphism $\psi \colon \WMSigma \rightarrow R$ with $\psi(\ext) = \ext$ recognises $\{a^nb^n \mid n \in \N\}$. 
\end{example}

\begin{example}\label{ex:ExtAlgebraHPlus}
As a second example, consider the language $H^+$ \cite{Ha15} over the alphabet $A = \{a,b\}$, where $a$ is a call and $b$ a return letter. This language is given by the production rules
\[S \rightarrow aNb \mid SS \mid \lambda, \ N \rightarrow aSb \mid NN \mid NS \mid SN.\]

Intuitively speaking, the language encodes all true Boolean formulae in the sense that the empty word is considered true, concatenation is conjunction and enclosing a word by $a$ and $b$ is negation. Then $H^+$ is recognised by the $\Ext$-algebra $R_{H^+}$ defined below.

\begin{tabular}[t]{c|cc}
$\cdot$ & 1 &  0\\ \hline
$1$ & 1 & 0\\
$0$ & $0$ & $0$
\end{tabular}
\hspace{1cm}
\begin{tabular}[t]{c|cc}
& 1 & 0\\ \hline
$\ext$ & $0$ & $1$\\
$\ext[a^2,b^2]$ & 1 & 0\\
$\ext[ab,\lambda]$ & 0 & 0\\
$\ext[a^2b,b]$ & 1 & 1\\
\end{tabular}
\end{example}

As an intermediate step towards the main theorem, we show that for each language of well-matched words, there exists a minimal $\Ext$-algebra recognising it. We say that an equivalence relation $\sim$ on an $\Ext$-algebra $R$ is a congruence of $\Ext$-algebras if and only if for all $e \in \cO(R)$: 
\[x \sim y \Leftrightarrow e(x) \sim e(y).\]
In the following, by congruence, we mean congruence of $\Ext$-algebras unless explicitly stated otherwise.

Observe that the map $x \mapsto [x]$, where $[x]$ denotes the equivalence class of $x$ with respect to $\sim$ is a morphism of $\Ext$-algebras, in the sense that its image is equipped with the operations $e([x]) = [e(x)]$ for each $e \in \cO(R)$. We denote that $\Ext$-algebra by $R /_\sim$.

\begin{definition}[Syntactic Congruence]
Let $L \subseteq \FreeExt$. We say that two words $x,y \in \FreeExt$ are equivalent with respect to $L$ and write $x \sim_L y$ if for all $e \in \cO(\WMSigma)$, 
$$e(x) \in L \Leftrightarrow e(y) \in L$$
holds. We call $\WMSigma /_{\sim_L}$ the syntactic $\Ext$-algebra of $L$ and the canonical map induced by $\sim_L$ the syntactic morphism $\eta_L$.
\end{definition}

Observe that for well-matched words $u$ and $v$, this implies that if $x \sim_L y$ then the equivalence $\ext[u,v](x) = uxv \in L \Leftrightarrow \ext[u,v](y) = uyv \in L$ holds, but the converse does not hold in general.

We may now characterise the syntactic $\Ext$-algebra of $L$ as the smallest $\Ext$-algebra recognising $L$ with respect to subs and quotients: 

\begin{definition}[Subs and Quotients]
Let $R$ and $S$ be $\Ext$-algebras, then 
\begin{itemize}
\item $R$ is a sub of $S$, if $R \subseteq S$ and $\cO(R) \subseteq \cO(S)$.
\item $R$ is a quotient of $S$, if there exists a morphism $(\phi,\psi)$ from $S$ to $R$, such that $\phi$ and $\psi$ are surjective.
\end{itemize}
\end{definition}

\begin{proposition}\label{prop:RecognitionDivision}
An $\Ext$-algebra $R$ recognises a VPL $L$ if and only if the syntactic $\Ext$-algebra of $L$ is the quotient of a sub of $R$.
\end{proposition}

The proof is similar as the proof for the statement that a finite monoid recognises a certain (regular) language if and only if its syntactic monoid divides it. With the previous proposition, we are now ready to come to the main theorem of the section:

\begin{theorem}\label{thm:VPL-Ext}
A language $L \subseteq \FreeExt$ is VPL, if and only if it is recognised by a finite $\Ext$-algebra.
\end{theorem}

The proof is similar to that in \cite{AKMV05} enriched by $\ext[u,v]$-operations.

\textbf{Proof sketch:} If $\mathcal A$ is a VPA recognising a VPL $L$, then each pair $(w,G)$, where $w$ is a well-matched word and $G$ is some stack symbol, induces a map $f_{w,G}$ from the states of $\mathcal A$ to the states of $\mathcal A$. One shows that the relation $u \sim v$ if $f_{u,G} = f_{v,G}$ for all $G$ is a finite $\Ext$-algebra congruence on $\FreeExt$. Furthermore, this congruence refines the syntactic congruence and hence by Proposition \ref{prop:RecognitionDivision}, $L$ is recognised by the finite $\Ext$-algebra $\WMSigma /_\sim$. The number of elements of that $\Ext$-algebra can be exponential in the number of states of $\mathcal A$.

For the converse direction, if $R$ is an $\Ext$-algebra recognising $L$, one constructs a VPA with the elements of $R$ as states. The automaton then simulates the evaluation of a word $w$ in $R$ by keeping track of the read call-letters, pushing them on the stack and when reading a return letter, applying the appropriate $\ext$-operation to the state it is currently in.\\

Observe that the sketched procedure above describes an algorithm to calculate an $\Ext$-algebra recognising the same language $L$ as a given VPA $\cA$. This is possible, since there are finitely many functions $g$ from the states of $\cA$ to its states and hence for each such function there exists a $u \in \WMSigma$ and a stack symbol $G$, such that $g = f_{u,G}$. Moreover, the length of $u$ is bounded by a constant dependant only on the number of states of $\cA$. Hence, the representatives $u$ for the equivalence classes can be found in finite time. One may then construct the syntactic $\Ext$-algebra of $L$ by a standard minimisation procedure, merging equivalence classes.

\section{Proofs of Section \ref{sec:Algebra}}

We say that an $\Ext$-algebra $R$ divides an $\Ext$-algebra $S$, if $R$ is the quotient of a sub of $S$. Moreover, the syntactic $\Ext$-algebra of a languages $L \subseteq \WMSigma$ is denoted by $\Ext(L)$.

\subsection{Proof of Proposition \ref{prop:RecognitionDivision}:}
\begin{proof}
Let us first prove the converse direction and suppose that $S$ is some $\Ext$-algebra that is
divided by $\Ext(L)$. Hence there exists a subalgebra $T$ of $S$ and a surjective morphism
$\pi \colon T \rightarrow \Ext(L)$. By $\eta_L$ denote the syntactic morphism of $L$. We show
the existence of a morphism $h \colon \WMSigma \rightarrow S$ such that the following diagram
commutes. 
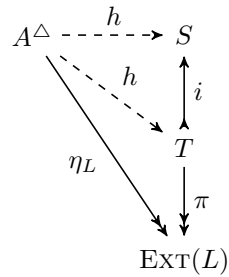
\begin{center}
\begin{tikzpicture}[->,>=stealth',shorten >=1pt,auto,node distance=3.4cm,semithick,scale=1]
\node(W) at (0,0) {$\WMSigma$};
\node(S) at (2,0) {$S$};
\node(T) at (2,-1.5) {$T$};
\node(E) at (2,-3) {$\Ext(L)$};

\draw[->,dashed] (W) to node[above] {$h$} (S);
\draw[->,dashed] (W) to node[above right] {$h$} (T);
\draw[->>] (W) to node[below left] {$\eta_L$} (E);
\draw[>->] (T) to node[right] {$i$} (S);
\draw[->>] (T) to node[right] {$\pi$} (E);
\end{tikzpicture}
\end{center}
Define the function $h \colon \Intern \rightarrow T$, by choosing $h(a) \in \pi^{-1}(\eta_L(a))$.
Since $\pi$ is surjective, such an element exists. By the unique property of the free
$\Ext$-algebra $\WMSigma$, $h$ naturally extends to a morphism $h \colon \WMSigma \rightarrow
T$ and since $T$ is a subalgebra of $S$, we can view $h$ as a morphism from $\WMSigma$ to $S$. By
the choice of $h$, we have for $w \in \WMSigma$
$$ \pi(h(w)) = \prod_{i=1}^{|w|} \pi(h(w_i)) = \prod_{i=1}^{|w|} \eta(w_i) = \eta(w)$$
and hence, by $L = \eta^{-1}(\eta(L))$
$$w \in L \Leftrightarrow \eta(w) \in \eta(L)
	\Leftrightarrow \pi(h(w)) \in \eta(L)
	\Leftrightarrow h(w) \in \pi^{-1}(\eta(L)).
$$
Thus $L = h^{-1}(\pi^{-1}(\eta(L)))$ and $S$ recognises $L$.

Now assume that $S$ is some $\Ext$-algebra that recognises $L$ by the morphism $h \colon
\WMSigma \rightarrow S$. Then the image of $h$ is a subalgebra of $S$. We show that $h(S)$ has
$\Ext(L)$ as a quotient. For that, we prove that $\eta$ factors through $h$, that is for any
two $u,v \in \WMSigma$, $h(u) = h(v)$ implies $\eta(u) = \eta(v)$.
If $h(u) = h(v)$, since $h$ is a morphism recognising $L$, we
obtain $xuy \in L$ iff $xvy \in L$ and $\ext(u) \in L$ iff $\ext(v) \in L$ and hence
$\eta(u) = \eta(v)$. Thus we can define $\pi \colon h(S) \rightarrow \Ext(L)$ by letting
$\pi(s) = \eta(h^{-1}(s))$. One can then verify that $\pi$ is indeed a morphism and thus
$\Ext(L)$ is a quotient of $S$.
\end{proof}

\subsection{Proof of Theorem \ref{thm:VPL-Ext}:}

In the proof, we use the following notation for VPA:

\begin{definition}[Visibly pushdown automaton (VPA)]
A \emph{visibly pushdown automaton} is a tuple $(\VPSigma,Q,q_0,\Gamma,\#,\delta,F)$, where
\begin{itemize}
\item $\VPSigma$ is a visibly pushdown alphabet,
\item $Q$ is a finite set, the set of \emph{states},
\item $q_0 \in Q$ is the \emph{initial state},
\item $\Gamma$ is a finite alphabet, the \emph{stack alphabet},
\item $\# \in \Gamma$ is the \emph{bottom-of-stack symbol},
\item $\delta \colon A \times Q \times \Gamma \rightarrow Q \times \Gamma^*$ is the transition function with the following restrictions:
For $q \in Q, a \in A$ and 
$$\delta(a,q,G) = (q',G'),$$
where $G \in \Gamma$ and $G' \in (\Gamma \backslash \{\#\})^*$, it must hold that
	\begin{itemize}
		\item if $a \in \Call$, then $G' = G_0 G$, for some $G_0 \in \Gamma \backslash \{\#\}$.
		\item if $a \in \Return$, then $G' = \emptyword$ .
		\item if $a \in \Intern$, then $G' = G$.
	\end{itemize}
\item and $F \subseteq Q$ is the set of final states.
\end{itemize}
\end{definition}
\begin{definition}[Language of a VPA]
Let $M = (A,Q,q_0,\Gamma,\#,\delta,F)$ be a visibly pushdown automaton and let $k \in \N$ and $G_i \in \Gamma$ for $i = 1,\dots,k$. We define the \emph{extended transition function}, denoted by $\widehat{\delta} \colon A^* \times Q \times \Gamma^* \rightarrow Q \times \Gamma^*$, inductively as
\begin{itemize}
\item $\widehat{\delta}(\emptyword,q,G_0 \dots G_k) = (q,G_0 \dots G_k)$
\item $\widehat{\delta}(aw,q,G_0 \dots G_k) = \widehat{\delta}(w,q',G' G_1 \dots G_k)$, where $(q',G') = \delta(a,q,G_0)$. 
\end{itemize}
The language accepted by $M$ is the language
$$L(M) = \{w \in A^* \mid \widehat{\delta}(w,q_0,\#) \in F \times \{\#\}\}$$
\end{definition}

\begin{lemma}\label{lem:stackBehav}
Let $w \in \WMSigma$ and let $M = (A,Q,q_0,\Gamma,\#,\delta,F)$ be a VPA. Then
\[\what{\delta}(w,q,G) = \pi_Q(\what{\delta}(w,q,G)) \times \{G\}\]
\end{lemma}

The proof of this Lemma is entirely straight-forward.
\begin{proof}
Let $L \subseteq \FreeExt$ be a VPL and let $M_L = (A,Q,q_0,\Gamma,\#,\delta,F)$ be a VPA accepting $L$. Recall that by $\pi_Q \colon Q \times \Gamma^* \rightarrow Q$, we denote the projection to the state. We will now define an equivalence on well matched words, based on the states of the automaton $M_L$. Let $w \in \FreeExt, G \in \Gamma$ and define the function
\begin{align*}
f_{w,G} \colon Q &\rightarrow Q\\
q 	     &\mapsto \pi_Q(\widehat{\delta}(w,q,G)).
\end{align*}
Observe that for $w_1,w_2 \in \FreeExt$ the relation
$$w_1 \sim_{M_L} w_2 \text{ iff for all $G \in \Gamma$, } f_{w_1,G} = f_{w_2,G}$$
is an equivalence relation on $\FreeExt$. Note that it is also a congruence on $\FreeExt$, since for $w_1,w_2,z \in \FreeExt$ with $w_1 \sim_{M_L} w_2$ and $G \in \Gamma$, we have
\begin{align*}
\pi_Q(\widehat{\delta}(zw_1,q,G)) &= \pi_Q(\widehat{\delta}(w_1,\widehat{\delta}(z,q,G)))\\
				 &= \pi_Q(\widehat{\delta}(w_1,\widehat{\delta}(\emptyword,q',G))) \text{ for some $q'$, since $z$ is well-matched,}\\
				 &= \pi_Q(\widehat{\delta}(w_1,q',G))\\
				 &= \pi_Q(\widehat{\delta}(w_2,q',G))\text{ since $f_{w_1,G} = f_{w_2,G}$,}\\
				 &= \pi_Q(\widehat{\delta}(w_2,\widehat{\delta}(z,q,G)))\\ 
				 &= \pi_Q(\widehat{\delta}(zw_2,q,G)).\\ 
\end{align*}
The case for $w_1z$ and $w_2z$ follows from Lemma \ref{lem:stackBehav}. Combining the two yields $xw_1y \sim_{M_L} xw_2y$ for $x,y \in \FreeExt$. Moreover
\begin{align*}
\pi_Q(\widehat{\delta}(aw_1b,q,G)) &= \pi_Q(\widehat{\delta}(w_1b,\delta(a,q,G)))\\
					     &= \pi_Q(\widehat{\delta}(w_1b,q',G_aG)) \text{ for some $q' \in Q, G_a \in \Gamma$,}\\
					     &= \pi_Q(\widehat{\delta}(b,\pi_Q(\widehat{\delta}(w_1,q',G_a)),G_aG)) \text{ by Lemma \ref{lem:stackBehav},}\\
					     &= \pi_Q(\widehat{\delta}(b,\pi_Q(\widehat{\delta}(w_2,q',G_a)),G_aG)) \text{ since $f_{w_1,G_a} = f_{w_2,G_a}$,}\\
					     &= \pi_Q(\widehat{\delta}(w_2b,q',G_aG))\\
					     &= \pi_Q(\widehat{\delta}(w_2b,\delta(a,q,G)))\\
					     &= \pi_Q(\widehat{\delta}(\ext(w_2),q,G))\\
\end{align*}
Since $\Gamma$ is finite and there are $|Q|^{|Q|}$ different functions from $Q$ to $Q$, $\sim_{M_L}$ has at most $|\Gamma| \cdot |Q|^{|Q|}$ congruence classes and $\sim_{M_L}$ is finite, which also makes $\WMSigma \backslash_{\sim{M_L}}$ finite. By construction if $w_1 \sim_{M_L} w_2$, then $f_{w_1,\#}(q_0) = f_{w_2,\#}(q_0)$ and hence
$$w_1 \in L \Leftrightarrow \pi_Q(\widehat{\delta}(w_1,q_0,\#)) \in F \Leftrightarrow \pi_Q(\widehat{\delta}(w_2,q_0,\#)) \in F \Leftrightarrow w_2 \in L,$$
which implies that $\sim_{M_L}$ is a refinement of the syntactic congruence, which in turn implies that the syntactic $\Ext$-algebra of $L$ divides $\FreeExt /_{\sim_{M_L}}$. Thus $L$ is recognised by a finite $\Ext$-algebra by Proposition \ref{prop:RecognitionDivision}.

For the converse direction, assume that $L$ is recognised by a finite $\Ext$-algebra $R$ via a morphism $h \colon \Ext\rightarrow R$. We construct a visibly pushdown automaton $M = (A,Q,q_0,\Gamma,\#,\delta,F)$ recognising $L$ as follows
\begin{itemize}
\item $Q = R$
\item $q_0 = h(\lambda)$
\item $\Gamma = \{\#\} \cup (Q \times \Call)$
\item $\delta \colon A \times Q \times \Gamma \rightarrow Q \times \Gamma^*$ is defined as follows: Let $a,a' \in A$ and $q,q' \in Q$. Then
\[
\delta(a,q,(q',a')) = 
\begin{cases}
(h(\emptyword),(q,a)(q',a')) & \text{ if $a \in \Call$,}\\
 (q' \cdot \ext[a'a](q),\emptyword) & \text{ if $a \in \Return$ and}\\
 (q \cdot h(a),(q',a')) & \text{ if $a \in \Intern$}\\
\end{cases}
\]
\item $F = h(L)$.
\end{itemize}
We have to show that for each $w \in \FreeExt$, $\pi_Q(\widehat{\delta}(w,q_0,\#)) \in F$ if and only if $w \in L$. We prove this by showing that $\widehat{\delta}(w,q,G) = (q \cdot h(w),G)$ for each $G \in \Gamma^*$ by induction on the structure of words in $\FreeExt$.

\textbf{Inductive start:}
Let $w \in \Intern$ then $$\widehat{\delta}(w,q,G) = \delta(w,q,G) = (h(\emptyword) \cdot h(w),G).$$

\textbf{Inductive step:} Let $w,w_1,w_2 \in \FreeExt$ be some words for which the claim holds.
Then
\begin{align*}
\widehat{\delta}(w_1 \cdot w_2,q,G) &=\widehat{\delta}(w_2,\widehat{\delta}(w_1,q,G)\\
					     &= \widehat{\delta}(w_2,q \cdot h(w_1),G)\\
					     &= (q \cdot h(w_1) \cdot h(w_2),G)\\
					     &= (h(w_1 \cdot w_2),G)
\end{align*}
and
\begin{align*}
\widehat{\delta}(awb,q,G) &=\widehat{\delta}(wb,\delta(a,q,G))\\
					     &=\widehat{\delta}(wb,h(\emptyword),(q,a)G)\\
					     &=\widehat{\delta}(b,\pi_Q(\widehat{\delta}(w,h(\emptyword),(q,a))),G)\\
					     &=\widehat{\delta}(b,h(w),(q,a)G)\\
					     &=\delta(b,h(w),(q,a)G)\\
					     &= (q \cdot \ext(h(w)),G)\\
					     &= (q \cdot h(awb),G)
\end{align*}
Setting $q = q_0 = h(\emptyword)$ proves the claim.
\end{proof}

\section{Varieties and an Eilenberg Theorem}\label{sec:Eilenberg}

In this section, we show that there is a one-to-one correspondence between classes of VPL and classes of $\Ext$-algebras with certain closure properties.\\

\textbf{Note.} To readers familiar with universal algebra, these closure properties should come as no surprise, see pseudo-varieties in \cite{Al95}. However, to keep the subject accessible to a broader community, we refrain from using the slang of universal algebra. It should however be mentioned, that it might be possible to obtain these results using category theoretic machinery as, for instance, in \cite{ACMU16} or \cite{Bo15}.\\

We define the following operations on well-matched words: If $L \subseteq \FreeExt$ is a language of well-matched words and $u,v \in A^*$ words such that $uv \in \WMSigma$, then
\[\ext[u,v]^\mone(L) = \{w \in \FreeExt \mid \ext[u,v](w) \in L\},\]
which we call an inverse $\Ext$-operation. Observe that, for instance, $\ext[w,\lambda]^\mone(L)$ for $w \in \WMSigma$ is very similar to what is known under the name of quotients by words for languages over $A^*$.

\begin{definition}[Pseudo-Variety of VPL]
A pseudo-variety of visibly pushdown languages is a class $\cV$ of languages of well-matched words such that
\begin{enumerate}
\item for each visibly pushdown alphabet $\VPSigma$, the set $\cV(\VPSigma)$ is a Boolean algebra of VPL over $\FreeExt$,
		\item the set $\cV(\VPSigma)$ is closed under inverse extend operations, that is for $L \in \cV(\VPSigma)$ and $u,v \in A^*$ such that $uv \in \WMSigma$, $\ext[u,v]^\mone(L)$ is an element of $\cV(\VPSigma)$,
		\item and $\cV$ is closed under inverse morphisms, that is if $\psi \colon \FreeExt \rightarrow \WMSigma[B]$ is a morphism of $\Ext$-algebras, then $L \in \cV(\VPSigma[B])$ implies $\psi^\mone(L) \in \cV(\VPSigma)$.
\end{enumerate}
\end{definition}

To define the closure properties of the corresponding classes of $\Ext$-algebras, we need the notion of (finite) products: 

\begin{definition}[Direct Product]
  Let $R$ and $S$ be two $\Ext$-algebras, then their direct product $R \times S$ is the monoid $R \times S$ with component wise multiplication and $\cO(R \times S)$ is generated by the maps 
  \[f_R(r,s) = (e_R(r),s) \text{ and } f_S(r,s) = (r,e_S(s)) \text{ for } e_R \in \cO(R) \text{ and } e_S \in \cO(S).\]
\end{definition}

\begin{definition}[Peudo-Variety of $\Ext$-algebras]
A class $\bV$ of $\Ext$-algebras is a pseudo-variety, if it is closed under
\begin{itemize}
\item subs, that is if $S \in \bV$ and $R$ is a sub of $S$, then $R \in \bV$,
\item quotients, that is if $S \in \bV$ and $R$ is a quotient of $S$, then $R \in \bV$ and 
\item finite direct products, that is if $R,S \in \bV$ then $R \times S \in \bV$.
\end{itemize}
\end{definition}

\begin{theorem}\label{thm:eil}
There is a one-to-one correspondence between varieties of VPL and varieties of $\Ext$-algebras.
\end{theorem}

\textbf{Proof Sketch:} We define the correspondence $\bV \rightarrow \cV$ by sending a pseudo-variety of $\Ext$-algebras $\bV$ to the class of all languages recognised by members of $\bV$. This class turns out to be a pseudo-variety of VPL. Conversely, we define the correspondence $\cV \rightarrow \bV$, where a pseudo-variety of VPL is sent to the pseudo-variety generated by all syntactic $\Ext$-algebras. One then shows that these correspondences are mutually inverse bijections. The constructions are similar to those in the proof of the original Eilenberg theorem (or see for instance \cite{St02}).

\section{Proofs of Section \ref{sec:Eilenberg}}
\subsection{Proof of Theorem \ref{thm:eil}:}

The proofs following from here up to the proofs of Section \ref{sec:application} are almost entirely along the lines of the book of Jean-\'{E}ric Pin - Mathematical Foundations of Automata Theory \footnote{Version November 30, 2016: \texttt{https://www.irif.fr/~jep/PDF/MPRI/MPRI.pdf}}. The parts for the monoid component of $\Ext$-algebras are very similar to proofs for finite monoids, the $\ext$ operations are added.

Since the following proofs do not differ significantly from the ones that give the original Eilenberg theorem between varieties of regular languages and varieties of finite monoids, we keep them short.

\begin{proposition}\label{prop:VarietyRecAtoms}
  Let $\cV$ be a variety of VPL, $L \in \cV(\VPSigma)$ and $\eta_L \colon \FreeExt \rightarrow R$ the syntactic morphism. Then for each $x \in R$, $\eta^\mone(x) \in \cV(\VPSigma)$.
\end{proposition}

\begin{proof}
  Let $K$ be the syntactic image of $L$, that is $K = \eta_L(L)$. Then $L = \eta_L^\mone(K)$. We prove that we can express $x$ as a Boolean combination of quotients and inverse extend operations of $K$.
  Define the sets
  \[C = \bigcap_{\substack{s,t \in \FreeExt\\ x \in s^\mone K t^\mone}} s^\mone K t^\mone \cap \bigcap_{\substack{(a,b) \in A_C \times A_R \\ x \in \ext^\mone(K)}} \ext^\mone(K)\]
  and
  \[N = \bigcup_{\substack{s,t \in \FreeExt\\ x \notin s^\mone K t^\mone}} s^\mone K t^\mone \cup \bigcup_{\substack{(a,b) \in A_C \times A_R \\ x \notin \ext^\mone(K)}} \ext^\mone(K).\]
  Then $u \in C \backslash N$ if and only if $u \sim_K x$, where $\sim_K$ is the syntactic congruence of $K$ on $R$. Since $R$ is the syntactic $\Ext$-algebra of $L$ and $K = \eta(L)$, it follows that
  \[C \backslash N = \{x\}.\]
  Since $\eta^\mone(K) \in \cV(\VPSigma)$ and the preimage of a morphism commutes both with quotients by words and inverse extend operations, we obtain $\eta^\mone(\{x\} \in \cV(A^*)$.
\end{proof}

\begin{proposition}\label{prop:ExtVarFormsLangVar}
   Let $\bV$ be a variety of $\Ext$-algebras, then the languages recognised by $\bV$ form a variety of well-matched languages.
\end{proposition}

\begin{proof}
  Denote by $\cV(\VPSigma)$ the set of all languages over $\FreeExt$ that are recognised by elements of $\bV$. Since $\bV$ is closed under direct products, it follows that $\cV(\VPSigma)$ is a Boolean algebra.

  Let $R \in \bV$ and $h \colon \FreeExt \rightarrow R$ a morphism with $L = h^\mone(K)$ for some $K \subseteq R$. It is straight-forward that for any $w \in \FreeExt$
  \[w^\mone L = h^\mone\left( \left\{ m \in R \mid h(w) m \in K  \right\} \right)\]
  and hence $w^\mone L \in \cV(\VPSigma)$. It follows from a similar argument, that also $w^\mone L \in \cV(\VPSigma)$ and $\ext^\mone(L) \in \cV(\VPSigma)$.

  Also, $\cV$ is closed under inverse morphisms, for the reason that if $h \colon \FreeExt[B] \rightarrow R$ recognises $L$, then $h \circ \varphi$ with $\varphi \colon \FreeExt \rightarrow \WMSigma[B]$ recognises $\varphi^\mone(L)$.

\end{proof}

We denote the correspondence sending a variety of $\Ext$-algebras to a variety of VPL by $\bV \rightarrow \cV$, where a variety of $\Ext$-algebras maps to the variety of all languages recognised by members of $\bV$.

\begin{proposition}
  The correspondence $\bV \rightarrow \cV$ is one-to-one.
\end{proposition}

\begin{proof}
  Assume that $\bV$ and $\bW$ are two varieties of $\Ext$-algebras with $\bV \mapsto \cV$ and $\bW \mapsto \cV$. For an $\Ext$-algebra $R \in \bV$ and morphism $h \colon \FreeExt \rightarrow R$ and any $m \in R$, define the language $L_m = h^\mone(m)$. Observe that by \ref{prop:RecognitionDivision}, the syntactic monoid $\Ext(L_m)$ is contained in $\bV$ and $\bW$. Also, $R$ divides
  \[\prod_{m \in R} \Ext(L_m),\]
which results in $R \in \bW$. By the symmetry of the argument $\bV = \bW$.
\end{proof}

To each variety $\cV$ of VPL, associate the variety of $\Ext$-algebras generated by all syntactic $\Ext$-algebras of languages $L \in \cV(\VPSigma)$ for some visibly pushdown alphabet $\VPSigma$. Denote that correspondence by $\cV \rightarrow \bV$.

\begin{theorem}
  The correspondences $\cV \rightarrow \bV$ and $\bV \rightarrow \cV$ are mutually inverse bijections.
\end{theorem}

\begin{proof}
  Let $\cV$ be a variety of VPL with $\cV \mapsto \bV \mapsto \cW$. We show that $\cV = \cW$.

  Let $L \in \cV(\VPSigma)$, then the syntactic $\Ext$-algebra of $L$ is contained in $\bV$ and hence also $L \in \cW(\VPSigma)$.

  For the converse direction assume that $L \in \cW(\VPSigma)$. By Proposition \ref{prop:RecognitionDivision} the syntactic $\Ext$-algebra of $L$, denoted by $R_L$ is contained in $\bV$. Since $\bV$ is the variety of $\Ext$-algebras generated by all syntactic $\Ext$-algebras of languages of $\cV$, there exist an $n \in \N$ and for $i = 1,\dots,n$ visibly pushdown alphabets $\VPSigma_i$ and languages $L_i \subseteq \WMSigma[A_i]$ such that $R_L$ divides the product
  \[R := \prod_{i=1}^n R_{L_i}.\]
  It follows immediately that also $R$ recognises $L$ and we denote the morphism recognising $L$ by $\varphi \colon \FreeExt \rightarrow R$. Denote by $\pi_i \colon R \rightarrow R_{L_i}$ the projection on the $i$th component and by $\varphi_i = \pi_i \circ \varphi$. Then there exist morphisms $\psi_i \colon \FreeExt \rightarrow \Ext(\VPSigma_i)$ such that the diagram
  \begin{center}
  \begin{tikzpicture}[->,>=stealth',shorten >=1pt,auto,node distance=3.4cm, scale=1]
		\node(A) at (0,0) {$\FreeExt$};
	    \node(Ai) at (2.7,0) {$\WMSigma[A_i]$};
	    \node(R) at (0,-2.7) {$R$};
		\node(Ri) at (2.7,-2.7) {$R_{L_i}$};
		\draw[->] (A) to node[left] {$\varphi$} (R);
		\draw[->] (A) to node[above] {$\psi_i$} (Ai);
		\draw[->] (A) to node[above right] {$\varphi_i$} (Ri);
		\draw[->] (Ai) to node[right] {$\eta_{L_i}$} (Ri);
		\draw[->] (R) to node[below] {$\pi_i$} (Ri);
  \end{tikzpicture}
\end{center}
  commutes.
  Observe that since $R$ recognises $L$, there exists some $P \subseteq R$ such that
  \[L = \varphi^\mone(K) = \bigcup_{x \in K} \varphi^\mone(x)\]
  and letting $x = (x_1,\dots,x_n)$
  \[\varphi^\mone(x) = \bigcap_{i=1}^n \varphi_i^\mone(x_i).\]
  From the previous diagram, we get $\varphi_i = (\eta_{L_i} \circ \psi_i)$. We conclude that $L \in \cV(\VPSigma)$: Since $\cV(\VPSigma)$ is closed under Boolean combinations and inverse morphisms, it suffices to that $\eta_{L_i}^\mone(x_i) \in \cV(\VPSigma_i)$, which follows directly from Proposition \ref{prop:VarietyRecAtoms}.
\end{proof}

\section{The Free Profinite $\Ext$-algebra}\label{sec:topology}

Constructing the free profinite $\Ext$-algebra is the key ingredient to obtaining a notion of equations for pseudo-varieties of $\Ext$-algebras in the sense that a language belongs to a pseudo-variety if and only if it satisfies all of its equations. 

We are going to use the following approach: Starting with a notion of metric on the set of well-matched words, we are going to consider its metric completion $\ExtHat$. This initially provides us with a topological space, which contains the well-matched words as a subset. We are then going to show that $\ExtHat$ can be equipped with algebraic structure, such that it is an $\Ext$-algebra. To be more precise, that it can be equipped with a multiplication which agrees with the usual multiplication on the subset of the well-matched words and that (for $u,v \in A^*$ such that $uv \in \WMSigma$) for each $\ext[u,v] \in \cO(\WMSigma)$, there exists a map $\what{\ext[u,v]} \in \cO(\ExtHat)$ agreeing with it on the well-matched words. 

While the construction of the metric completion and the fact that it has a multiplication are almost entirely the same as in the well-known finite monoid (or regular) case \cite{Pi16}, the construction of $\cO(\ExtHat)$ requires some extra consideration. 

\textbf{Note.} The concluding theorem of this section is a special case of Reiterman's theorem \cite{Al95}, which holds for arbitrary pseudo-varieties of finitary algebras (or could probably be obtained using \cite{ACMU15}). We review the precise construction in this chapter, since it helps to understand the interpretation of equations on finite $\Ext$-algebras. For reading on free profinite forest algebras, see \cite{Ali16}.

\subsection{The Well-matched Words as a Metric Space}

We say that an $\Ext$-algebra $R$ separates two well-matched words $x,y \in \FreeExt$, if there is a morphism $\psi \colon \FreeExt \rightarrow R$ such that $\psi(x) \neq \psi(y)$. 

\begin{example}\label{ex:separation}
Let $a,c$ be call and $b,d$ be return letters, then $a^4b^2c^2d^4 \text{ and } a^2b^2c^2d^2$ are separated by the $\Ext$-algebra 
\begin{center}
\begin{tabular}[t]{c|cc}
$\cdot$ & $1$ & $0$  \\ \hline
$1$ & $1$ & $0$ \\ 
$0$ & $0$ & $0$ \\ 
\end{tabular}
\hspace{1cm}
\begin{tabular}[t]{c|cc}
& $1$ & $0$ \\ \hline
$\ext, \ext[c,d]$ & $1$ & $0$\\ 
$\ext[a,d], \ext[c,b]$ & $0$ & $0$\\
\end{tabular}
\end{center}
since the word $a^4b^2c^2d^4 = \ext[a,d]^2(\ext[a,b]^2(\lambda)\ext[c,d]^2(\lambda))$ is mapped to $0$ and the word $a^2b^2c^2d^2 = \ext[a,b]^2(\lambda)\ext[c,d]^2(\lambda)$ to $1$.
However, the two words can not be distinguished by a finite monoid of size $2$. 
\end{example}
In fact, it is possible to construct well-matched words for any natural number $n$, that can be distinguished by the $\Ext$-algebra from Example \ref{ex:separation}, but not by a monoid of size $n$. 

\begin{proposition}\label{lem:separate}
\begin{itemize}
\item Let 
\[r(x,y) = \min\{|R| \mid R \text{ is a finite $\Ext$-algebra that separates $x$ and $y$}\}.\]
 Then the map $d \colon \WMSigma \times \WMSigma \rightarrow [0,\infty)$ with $(x,y) \mapsto 2^{-r(u,v)}$ defines a metric on $\WMSigma$.
\end{itemize}
\end{proposition}

Since $\WMSigma$ as a metric space is discrete and thus not particularly revealing, we consider its metric completion $\ExtHat$, which has $\WMSigma$ as a dense subspace. Recall that the completion may be obtained constructively as the set of equivalence classes of Cauchy-sequences, which we use here. For a more in-depth presentation of the concepts needed, we refer for instance to \cite{Pi16}. 

\subsection{The metric completion as an $\Ext$-algebra}

As it is already the case for the free profinite monoid (see \cite{Pi16}), $\ExtHat$ contains elements which are not in $\WMSigma$, such as for instance for each $x \in \WMSigma$
\[x^\omega = \lim_{n \rightarrow \infty} x^{n!}\]
is one of those elements.
The proof that the sequence $(x^{n!})_{n \in \N}$ is a Cauchy-sequence in $\WMSigma$ and the limit on the right hand side thus exists in $\ExtHat$ is entirely the same as in the case of finite monoids and thus omitted. 

So far, $\ExtHat$ is just a topological space without any algebraic properties. The following Proposition states that there exists a multiplication on $\ExtHat$, making it a monoid, which on the elements of $\WMSigma$ agrees with the usual multiplication. A similar statement holds for the operations $\ext[u,v]$.

\textbf{Note.} One may show that $\ExtHat$ is indeed the Stone space of visibly pushdown languages and as such similar to the recognisers for Boolean algebras over $A^*$ proposed in \cite{GGP10}. One difference is that we are regarding the VPLs as a Boolean algebra over well-matched words and hence $\ExtHat$ remains a monoid, contrary to the recognisers in \cite{GGP10}, which are a monoid if and only if the associated Boolean algebra consists of regular languages entirely.

\begin{proposition}\label{prop:concUniform}
  \begin{enumerate}
	\item The multiplication $\cdot$ on $\WMSigma$ has a unique and uniformly continuous extension $\what{\cdot}$ on $\ExtHat$.
	\item For all $u,v \in A^*$ with $uv \in \WMSigma$, the maps $\ext[u,v]$ have unique uniformly continuous extensions $\what{\ext[u,v]} \colon \ExtHat \rightarrow \ExtHat$.
  \end{enumerate}
\end{proposition}

Intuitively, one may view the extensions as follows: if $(x_n)_{n \in \N}$ (resp. $(y_n)_{n \in \N}$) is a Cauchy-sequence and $x$ (resp. $y$) its limit in $\ExtHat$, then $x \cdot y$ is equal to the limit of $(x_n \cdot y_n)_{n \in \N}$ and $\what{\ext[u,v]}(x)$ is equal to the limit of the sequence $(\ext[u,v](x_n))_{n \in \N}$. However, apart from the maps $\what{\ext[u,v]} \colon \ExtHat \rightarrow \ExtHat$, it is possible to derive further maps from elements of $\cO(\WMSigma)$ in a more general fashion.  

\begin{proposition}\label{prop:ExtensionOperations}
  Let $(e_n)_{n \in \N}$ be a sequence of elements in $\cO(\WMSigma)$, such that $(e_n(x))_{n \in \N}$ is Cauchy for each $x \in \WMSigma$. Then the sequence $(e_n)_{n \in \N}$ uniquely determines a uniformly continuous map $e \colon \ExtHat \rightarrow \ExtHat$.
\end{proposition}

\textbf{Proof sketch:} Since $(e_n(x))_{n \in \N}$ is Cauchy for each $x \in \WMSigma$, the map $e$ sending $x$ to $\lim_{n \rightarrow \infty} e_n(x)$ is a well-defined map from $\WMSigma$ to $\ExtHat$. Moreover, it is uniformly continuous, since an $\Ext$-algebra that does not separate two well-matched words $x,y$ does also not separate $e_n(x)$ and $e_n(y)$ for each $n \in \N$. Hence $d(e(x),e(y)) \leq d(x,y)$, which implies uniform continuity and thus there exists a uniformly continuous extension $\what{e} \colon \ExtHat \rightarrow \ExtHat$.\\

The space $\ExtHat$ becomes an $\Ext$-algebra, with the uniformly continuous extension $\what{\cdot}$ of the multiplication on $\WMSigma$ as multiplication on $\ExtHat$ and the set $\cO(\ExtHat)$ is the set of all maps $e \colon \ExtHat \rightarrow \ExtHat$ obtained in the fashion of Proposition \ref{prop:ExtensionOperations}. 

Observe that $\cO(\ExtHat)$ indeed is a monoid, since if the element $e$ (resp. $f$) in $\cO(\ExtHat)$ is determined by the sequence $(e_n)_{n \in \N}$ (resp. $(f_n)_{n \in \N})$), then $e \circ f$ is determined by $(e_n \circ f_n)_{n \in \N}$. 

We now consider a sequence in $\cO(\ExtHat)$, which determines an operation in $\cO(\ExtHat)$ of particular importance for the equations in the following sections.

\begin{lemma}\label{lem:ExtCauchy}
Let $u,v \in A^*$ such that $uv \in \WMSigma$ and $x \in \WMSigma$, then the sequence $(\ext[u,v]^{n!}(x))_{n \in \N}$ is Cauchy. 
\end{lemma}

We denote by $\extomega[u,v] \colon \ExtHat \rightarrow \ExtHat$ the uniformly continuous extension of the map sending $x \in \WMSigma$ to $\lim_{n \rightarrow \infty} \ext[u,v]^{n!}(x)$.

From a first glance, since the well-matched words are a subset of $A^*$ with $\ext[u,v](w) = uwv$ (for appropriate $u,v$), it might seem convenient to write $\ext[u,v]^\omega(x) = u^\omega x v^\omega$. But this notation is very misleading for the following reason:

The profinite word $u^\omega$ is an idempotent in $\Ahat$ and thus for $u,v,u',v' \in A^*$, it holds that $u^\omega u^\omega v'^\omega u'^\omega v^\omega v^\omega =  u^\omega v'^\omega u'^\omega v^\omega$ in $\Ahat$.
But in $\ExtHat$, the element $\extomega[u,v](\extomega[u,v'](\lambda)\extomega[u',v](\lambda))$ does not equal $\extomega[u,v'](\lambda)\extomega[u',v](\lambda)$, as can be understood on the following Example \ref{ex:morphismProfinite}, letting $u=a,v'=b,u'=c$ and $v=d$.

\subsection{$\ExtHat$ and Morphisms}

Understanding the morphisms from $\ExtHat$ into finite $\Ext$-algebras is crucial for the interpretation of equations in later sections.

Observe that for each Cauchy sequence $(x_n)_{n \in \N}$ in $\WMSigma$ and each morphism $\psi \colon \WMSigma \rightarrow R$ into a finite $\Ext$-algebra $R$, the sequence $(\psi(x_n))_{n \in \N}$ becomes eventually stationary in $R$. For instance, it holds that for each morphism $\psi \colon \WMSigma \rightarrow R$ into a finite $\Ext$-algebra $R$, the sequence $(\psi(x^{n!}))_{n \in \N}$ is eventually equal to the unique idempotent generated by $\psi(x)$ in $R$.

Each morphism $\psi \colon \WMSigma \rightarrow R$ has a unique continuous extension $\what{\psi} \colon \ExtHat \rightarrow R$ and for an element $x \in \ExtHat$, which is the limit of a Cauchy-sequence $(x_n)_{n \in \N}$ in $\WMSigma$, 
\[\what{\psi}(x) = \lim_{n \rightarrow \infty} \psi(x_n),\]
 which converges in $R$, since $(x_n)_{n \in \N}$ is a Cauchy-sequence with respect to the metric installed on $\WMSigma$.

\begin{example}\label{ex:morphismProfinite}
  Consider the $\Ext$-algebra from Example \ref{ex:separation} and let $\psi$ be the morphism mapping $\ext[u,v]$ in $\WMSigma$ to its respective counterpart in that algebra. Then
  \[\what{\psi}(\extomega[a,d](\lambda)) = \lim_{n \rightarrow \infty} \psi(\ext[a,d]^{n!}(\lambda)) = \lim_{n \rightarrow \infty} \ext[a,d]^{n!}(1) = \ext[a,d](1) = 0,\]
but $\what{\psi}(\extomega[a,b](\lambda)) = \what{\psi}(\extomega[c,d](\lambda)) = 1$.
\end{example}

\subsection{Reiterman}

Let $u,v \in \ExtHat$. We say that an $\Ext$-algebra $R$ satisfies the equation $u = v$ if and only if for each morphism $\phi \colon \WMSigma \rightarrow R$ the equality $\what{\phi}(u) = \what{\phi}(v)$ holds. A VPL $L$ satisfies $u = v$ if there exists an $\Ext$-algebra recognising $L$, that satisfies $u = v$. 

We say that a pseudo-variety of $\Ext$-algebras $\bV$ is defined by a set of equations $E$, if an $\Ext$-algebra belongs to $\bV$ if and only if it satisfies all equations of $E$.

\begin{theorem}[Reiterman]\label{thm:Reiterman}
A class of $\Ext$-algebras is a pseudo-variety if and only if it can be defined by a set of (possibly infinitely many) equations.
\end{theorem}

\textbf{Proof sketch:} Let $E$ be a set of equations over $\ExtHat$ and $[\![ E ]\!]$ the set of all $\Ext$-algebras satisfying all equations in $E$.
The proof idea is entirely the same as for the free profinite monoid: $\Eqat{E}$ is a pseudo-variety of $\Ext$-algebras, since equations are preserved by direct products, subs and quotients. And for the converse direction, one shows that if $E$ is the set of equations satisfied by a pseudo-variety $\bV$, then $\Eqat{E} \subseteq \bV$. Observe that the other inclusion is straight forward.

\section{Proofs of Section \ref{sec:topology}}

\subsection{Proof of Lemma \ref{lem:separate}:}

\begin{proof}
That $d$ is positive definite follows directly from the fact that any two distinct words $x,y$ may be separated by the syntactic $\Ext$-algebra of $\{x\}$ (resp. $\{y\}$). The symmetry of $d$ is clear, hence the strong triangle inequality remains to be shown.

Assume that $R$ is an $\Ext$-algebra separating $x$ and $y$ and that $z \in \WMSigma$. Then $R$ must separate $x$ and $z$ or $y$ and $z$, which results in $r(x,y) \geq \min\{r(x,z),r(z,y)\}$. And since $d(x,y) = 2^{-r(x,y)}$, we obtain that the strong triangle inequality holds.
\end{proof}

\subsection{Proof of Proposition \ref{prop:concUniform}:}
\begin{proof}
To show that the concatenation $\cdot \colon \FreeExt \times \FreeExt \rightarrow \FreeExt$ sending $(u,v)$ to $u \cdot v$, is uniformly continuous, it suffices to prove $d(uv,u'v') \leq d((u,v),(u',v'))$, where the metric on the right is that of the product. Observe that by the strong triangle inequality, for all $u,v,u'v' \in \FreeExt$
\begin{align*}
d(uv,u'v') \leq \max\{d(uv,uv'),d(uv',u'v')\}\\
\end{align*}
and, since an $\Ext$-algebra separating $uv$ and $uv'$ (respectively $uv'$ and $u'v'$) also has to separate $v$ and $v'$ (respectively $u$ and $u'$) we obtain
$$d(uv,u'v') \leq \max\{d(uv,uv'),d(uv',u'v')\} \leq \max\{d(v,v'),d(u,u')\}.$$
Hence the concatenation is uniformly continuous with respect to the product metric on $\FreeExt \times \FreeExt$.

To see that $\ext$ is uniformly continuous, observe that if an $\Ext$-algebra separates $\ext(u)$ and $\ext(v)$, then it must also separate $u$ and $v$ and hence
$$d(\ext(u),\ext(v)) \leq d(u,v),$$
which implies the uniform continuity of $\ext$.

Thus Concatenation and $\ext$ have unique uniformly continuous extensions.
\end{proof}

\subsection{Proof of Lemma \ref{lem:ExtCauchy}:}
\begin{proof}
Let $n,m,N \in N$. We show that for any $n,m \geq N$, the profinite well-matched words $\ext^{n!}(x)$ and $\ext^{m!}(x)$ cannot be separated by an $\Ext$-algebra of size at most $N$. Assume that $\varphi \colon \ExtHat \rightarrow E$ is a morphism, where $E$ is an $\Ext$-algebra with $|E| \leq N$. Since $E$ is finite, there exists an $r \in \N$ such that $\ext^{r}(x) = \ext^{r}(\ext^{r}(x)) = \ext^{2r}(x)$, with $r \leq N$. Since $m,n \geq N$, $r$ divides both $n!$ and $m!$ and hence $\ext^{n!}(x) = \ext^{r}(x) = \ext^{m!}(x)$. This shows that $\ext^{n!}(x)$ and $\ext^{m!}(x)$ cannot be separated by an $\Ext$-algebra of size at most $N$ and hence the sequence $\ext^{n!}(x)$ is a Cauchy-sequence.
\end{proof}

\subsection{Proof that $\ExtHat$ is the Stone Space of the VPL}

We show that the clopens of $\ExtHat$ are isomorphic to the VPL over $A$.

\begin{proposition}
The space $\ExtHat$ is compact.
\end{proposition}
\begin{proof}
For $u,v \in \FreeExt$, define the relation $\sim_n$ by $u \sim_n v$, if $u$ and $v$ cannot be separated by an $\Ext$-algebra of size at most $n$. It is evident that $u \sim_n v$ is reflexive and symmetrical. Let $x \in \FreeExt$ with $u \sim_n x$ and $x \sim_n v$. Since any monoid separating $u$ and $v$ must separate $u$ and $x$ or $v$ and $x$, transitivity follows and $\sim_n$ is an equivalence relation. More precisely, it is also of finite index, since there are only finitely many $\Ext$ algebra of size at most $n$. Since, the equivalence class of $u$ is precisely the open ball $B_{2^{-n}}(u)$ on $\FreeExt$, it is covered by finitely many open balls of radius $2^{-n}$. It follows that $\FreeExt$ is totally bounded and since $\ExtHat$ is its completion, using the theorem of Heine-Borel, it is compact.

\end{proof}

\begin{proposition}\label{prop:ClopenCong}
Let $R \subseteq \ExtHat$ and let $E_R$ be its syntactic $\Ext$-algebra. Then the following conditions are equivalent
\begin{enumerate}
	\item $R$ is clopen,
	\item the syntactic congruence of $R$ is a clopen subset of $\ExtHat \times \ExtHat$
	\item $E_R$ is finite and the syntactic morphism of $R$ is a continuous map.
\end{enumerate}
\end{proposition}

\begin{proof}
\textbf{1. to 2.:} Let $R$ be a clopen subset of $\ExtHat$. Define the sets
$$M = \bigcap_{u,v \in \ExtHat} ((u^{-1}Rv^{-1} \times u^{-1}Rv^{-1}) \cup (u^{-1}R^cv^{-1} \times u^{-1}R^cv^{-1}))$$
and
$$X = \bigcap_{a \in \Call, b \in \Return} ((\ext^{-1}(R) \times \ext^{-1}(R)) \cup (\ext^{-1}(R) \times \ext^{-1}(R)))$$
Then
$$\sim_R \ = M \cup X.$$
We prove that both $M$ and $X$ are closed. Observe that since $R$ is clopen and $u^{-1}Rv^{-1}$ is the inverse image of $x \mapsto uxv$, which is a continuous map, each of the sets $u^{-1}Rv^{-1}$ is closed and so is $u^\mone R^c v^\mone$. Since $\ExtHat \times \ExtHat$ has the product topology, $M$ is closed. 

By a similar argumentation, the set $X$ is also closed.

We prove that $\sim_R^c$ is closed, by showing that the limit of any convergent sequence in $\sim_R^c$ is contained in it. Let $(s_n,t_n)$ be a convergent sequence in $\sim_R^c$ with limit $(s,t)$. Since $s_n$ and $t_n$ are not equivalent, there exist sequences $u_n$ and $v_n$ in $\ExtHat$ such that $u_ns_nv_n \in R$ and $u_nt_nv_n \notin R$ or sequences $a_n$ in $\Call$ and $b_n$ in $\Return$ such that $\ext[a_nb_n](s_n) \in R$ and $\ext[a_nb_n](t_n) \notin R$.

In the first case, because $\ExtHat$ is compact, both $u_n$ and $v_n$ have convergent subsequences $u'_n$ and $v'_n$ with limits $u'$ and $v'$. Since $R$ is clopen and the multiplication is continuous $u'_n s_n v'_n$ converges to $u'sv' \in R$ and $u'_n t_n v'_n$ converges to $u'tv' \notin R$. Thus $(s,t) \in \sim_R^c$.

In the second case, since $\Call$ and $\Return$ are finite, and $\ext$ is continuous, there exist $a \in \Call$ and $b \in \Return$ such that $\ext[a_nb_n](s_n)$ converges to $\ext(s)$ and $\ext[a_nb_n](t_n)$ converges to $\ext(t)$. Since $R$ is clopen, $\ext(s) \in R$ and $\ext(t) \notin R$ and thus $(s,t) \in \sim_R^c$.

We conclude that $\sim_R^c$ is closed and thus $\sim_R$ is open.

\textbf{2. to 3.:} We show that for any $x \in \ExtHat$, the equivalence class of $x$ with respect to $\sim_R$ is open. It holds that for each $x \in \Ext$, $(x,x) \in \sim_R$. Since $\sim_R$ is open, there exists an open set $U \subseteq \ExtHat$ such that $(x,x) \in U \times U \subset \sim_R$. Moreover $U$ must be fully contained in the equivalence class of $x$, since otherwise there exist two non-equivalent elements $u,v \in U$ with $(u,v) \in U \times U$, which is a contradiction to $U \times U \subset \sim_R$. Since $x$ was arbitrary, this implies that the $x$-classes of $\sim_R$ are open and thus form an open partition of $\ExtHat$. By compactness of $\ExtHat$ this partition is finite, which implies that $E_R$ is finite and the syntactic morphism is continuous by the observation that each $x$-class is the preimage of a singleton in $E_R$ and open.

\textbf{3. to 1.:} Let $\eta \colon \ExtHat \rightarrow E_R$ be the syntactic morphism of $R$. Since $E_R$ is finite, it is discrete and each subset of $E_R$ is clopen and since $\eta^{-1}(\eta(R)) = R$, and $\eta$ is continuous, $R$ is clopen.
\end{proof}

\begin{proposition}\label{prop:clopenVPL}
If $L \subseteq \FreeExt$ is a language, then $L = \overline{L} \cap \FreeExt$ and the following conditions are equivalent:
\begin{enumerate}
	\item $L$ is a VPL
	\item $L = K \cap \FreeExt$ for some clopen set $K \subseteq \ExtHat$
	\item $\overline{L}$ is clopen in $\ExtHat$
	\item $\overline{L}$ is recognised by a continuous morphism $\varphi \colon \ExtHat \rightarrow E$, where $E$ is a finite $\Ext$-algebra.
\end{enumerate}
\end{proposition}

\begin{proof}
\textbf{1. to 2.:} Let $L$ be a VPL, then by Theorem \ref{thm:VPL-Ext}, there exists a finite $\Ext$-algebra recognising $L$. Let $\varphi \colon \FreeExt \rightarrow E$ be the morphism recognising $L$, that is $L = \varphi^\mone(R)$, where $R \subseteq E$. Moreover, let $K = \widehat{\varphi}^{-1}(R)$. Since $E$ is discrete, $R$ is a clopen set and since $\what{\varphi}$ is continuous, $K$ is clopen. Then $\varphi(w) = \what{\varphi}(w)$ for $w \in \FreeExt$, implies $L = \what{\varphi}^\mone(R) \cap \FreeExt = K \cap \FreeExt$.

\textbf{2. to 3.:} Assume that $L = K \cap \FreeExt$ for some clopen set $K \subseteq \ExtHat$. Since $\FreeExt$ is dense in $\ExtHat$, and since $K$ is open, it follows that $K \cap \FreeExt$ is dense in $K$. Thus $\overline{K \cap \FreeExt} = K = \overline{L} $ is clopen.

\textbf{3. to 4.:} See Proposition $\ref{prop:ClopenCong}$

\textbf{4. to 1.:} Assume that $\overline{L}$ is recognised by a morphism $\varphi \colon \ExtHat \rightarrow E$ into a finite $\Ext$-algebra $E$. Then $\overline{L} = \varphi^\mone(R)$ for some $R \subseteq E$. Let $\varphi_{\FreeExt} \colon \FreeExt \rightarrow E$ be the restriction of $\varphi$ to $\FreeExt$. Then 
$$L = \overline{L} \cap \FreeExt = \varphi^{-1}(R) \cap \FreeExt = \varphi_{\FreeExt}^\mone(R)$$ and by Theorem \ref{thm:VPL-Ext}, $L$ is a VPL.
\end{proof}

The statements of the next Proposition characterise the topological closures of VPLs.

\begin{proposition}\label{prop:VPLClosure1}
Let $L \subseteq \FreeExt$ be a VPL and let $u \in \ExtHat$. Then the following conditions are equivalent:
\begin{enumerate}
	\item $u \in \overline{L}$	
	\item $\what{\varphi}(u) \in \varphi(L)$, for all morphisms $\varphi$ from $\FreeExt$ onto a finite $\Ext$-algebra.
	\item $\what{\eta}(u) \in \eta(L)$, were $\eta$ is the syntactic morphism of $L$.
	\item $\what{\varphi}(u) \in \varphi(L)$, for some morphism $\varphi$ from $\FreeExt$ onto a finite $\Ext$-algebra recognising $L$.
\end{enumerate}
\end{proposition}

\begin{proof}
\textbf{1. to 2.:} Suppose $u \in \overline{L}$ and let $\varphi \colon \WMSigma \rightarrow R$ be a morphism onto a finite $\Ext$-algebra $R$. Since $\what{\varphi}$ is continuous and $R$ is discrete we have 
\[\what{\varphi}(\overline{L}) = \overline{\what{\varphi}(L)} = \what{\varphi}(L) = \varphi(L).\]

\textbf{2. to 3. and 3. to 4.:} Is trivial.

\textbf{4. to 1.:} Let $\varphi \colon \WMSigma \rightarrow R$ be a morphism onto a finite $\Ext$-algebra $R$. If $R$ recognises $L$, then $L = \varphi^\mone(\varphi(L))$, which together with continuity of $\what{\varphi}$ and discreteness of $R$ implies
\[\overline{L} = \overline{\varphi^\mone(\varphi(L))} = \what{\varphi}^\mone(\what{\varphi}(\overline{L})) = \what{\varphi}^\mone(\varphi(L)).\]
\end{proof}

The following theorem is a conclusion from the previous propositions, in particular \ref{prop:clopenVPL}. 
Denote by $\text{VPL}(A)$ the set of visibly pushdown languages over the visibly pushdown alphabet $A$ and by $\text{Clopen}(\ExtHat)$ the set of all clopen sets of $\ExtHat$. We now show that the maps $L \mapsto \overline{L}$ from VPL$(A)$ to Clopen$(\ExtHat)$ and $K \mapsto K \cap \WMSigma$, where $K$ is a clopen of $\ExtHat$ are inverse to each other.

\begin{proof}
That both maps are inverse to each other, follows directly from Proposition \ref{prop:clopenVPL}. It remains to be shown that both are morphisms of Boolean algebra, which is straight-forward for the map $K \mapsto \WMSigma \cap K$. That $L \mapsto \overline{L}$ is a morphism can be derived from the observation that $\overline{L} = \what{\eta}^\mone(\eta(L))$, where $\eta$ is the syntactic morphism of $L$ and the fact that closure and union commute.
\end{proof}

\subsection{Proof of Theorem \ref{thm:Reiterman}:}

If $E$ is a set of equations, we denote by $\Eqat{E}$ the class of all $\Ext$-algebras satisfying all equations in $E$.

\begin{proposition}\label{prop:VPLEqVar}
If $E$ is a set of equations over $\ExtHat$, then $\Eqat{E}$ is a variety of $\Ext$-algebras and $L(\Eqat{E})$ is the corresponding variety of languages.
\end{proposition}

\begin{proof}
First of all, observe that the intersection of varieties is a variety and that thus, without loss of generality, it suffices to show that $\Eqat{u = v}$ for some $u,v \in \ExtHat$ forms a variety.

It is pretty straight-forward to see, that quotients and direct products preserve identities, which results in $\Eqat{u = v}$ being a variety of $\Ext$-algebras and by the observation that $L(\Eqat{u = v})$ are the languages recognised by $\Eqat{u = v}$ the claim follows.
\end{proof}

Proof of the Theorem

\begin{proof}
One direction of the proof follows directly from Proposition \ref{prop:VPLEqVar}. For the other direction, let $\bV$ be a variety of $\Ext$-algebras and let $E$ be the set of identities that are satisfied by all elements of $\bV$. Now, it is evident that $\bV \subseteq \Eqat{E}$. We prove that the inclusion $\Eqat{E} \subseteq \bV$ holds. Let $R \in \Eqat{E}$ and $\varphi \colon \ExtHat \rightarrow R$ be a morphism. For a morphism $h \colon \WMSigma \rightarrow S$, where $S \in \bV$ define the set
\[N_h = \{(u,v) \in \ExtHat \times \ExtHat \mid \what{h}(u) \neq \what{h}(v)\}.\]
This set is open, since $\what{h}$ is continuous and $N_h$ is the preimage of the complement of the diagonal of $S \times S$. Formally, we also identify $E$ with a subset of $\ExtHat \times \ExtHat$, that is the set of all $(u,v)$ such that $u = v$ is an equation in $E$. We observe that if a tuple $(u,v)$ is not contained in any set $N_h$ for some morphism $h$, then $(u,v)$ must be in $E$. Denote by $\cF$ the set of all morphisms from $\WMSigma$ into an $\Ext$-algebra of $\bV$. It follows that
\[E \cup \bigcup_{h \in \cF} N_h\]
is a cover of $\ExtHat \times \ExtHat$.
Define the set
\[E_{\varphi} = \{(u,v) \in \ExtHat \times \ExtHat \mid \what{\varphi}(u) = \what{\varphi}(v)\}.\]
Clearly $E \subseteq E_{\varphi}$ and $E_{\varphi}$ is open. By the previous argumentation,
\[E_\varphi \cup \bigcup_{h \in \cF} N_h\]
is an open cover of $\ExtHat \times \ExtHat$ an hence there exists a finite subcover
\[E_\varphi \cup \bigcup_{h \in F} N_h.\]
In particular, if $\what{h}(u) = \what{h}(v)$ for all $h \in F$, then $\what{\varphi}(u) = \what{\varphi}(v)$ and hence $R$ is a divisor of the product $\prod_{h \in F} \what{h}(\ExtHat)$ and since this finite $\Ext$-algebra is in $\bV$, we obtain $R \in \bV$.
\end{proof}

\section{Equations for Subclasses of VPL}\label{sec:application}

Unless stated otherwise, by equation, we mean in the following always a formal equality between two elements in $\ExtHat$.

We are now going to investigate equations for two subclasses of the visibly pushdown languages: The visibly counter languages and a subset of these, namely the set of all VCL with threshold zero. Hence, we start with VCL by determining a set of equations (see Proposition \ref{prop:EquationsVCL}), which is sound, where soundness means that each VCL satisfies the equations. For VCL with threshold zero, we show that additional equations (see Proposition \ref{prop:EquationsZeroVCL}) hold. Moreover, we consider examples for which the equations can be utilised to show that these particular example languages are not a VCL (or not a VCL with threshold zero).

\subsection{Equations for VCL}

The following set of equations is sound, but we do not know whether it is complete. In other words, we know that whenever a language is a VCL, then it satisfies the equation, but we do not know, whether the converse holds. However, we conjecture that the equations are indeed also complete.

\begin{proposition}\label{prop:EquationsVCL}
Let $x \in \ExtHat$ and let $u,v,u',v' \in A^*$ such that $uv,u'v',uv',u'v \in \WMSigma$. Moreover, let $L$ be a VCL, then $L$ satisfies the equation
\begin{equation}\label{eq:VCL}
\extomega[u,v](\extomega[u',v'](x)) = \extomega[u,v](\extomega[u,v'](\extomega[u',v'](x))) = \extomega[u,v](\extomega[u',v](\extomega[u',v'](x)))
\end{equation}
\end{proposition}

\textbf{Proof Sketch:} Assume that $L$ is a language recognised by a VCA $\cA$ with threshold $m$. Then for each counter value $i \leq m$ and (appropriate) word $u$ we obtain a function $r_{u,i}$ from the states of $\cA$ to the states of $\cA$, simply by simulating the run of $\cA$ on $u$ starting at stack height $i$ (appropriate $u$, since in the run the counter may not go below zero). By rather combinatorial arguments, it is possible to show that since $r_{u,k} = r_{u,l}$ for all $l,k \geq m$, there exists an exponent $s \in \N$ such that $r_{u^s,i} = r_{u^{2s},i}$ for all $i \in \N$.

We then derive that for words $u,v,u',v' \in A^*$ such that $uv,u'v',u'v,uv'$ are well-matched, there exists a common exponent $s$, such that
\[r_{u^su'^sxv'^sv^s,i} = r_{u^{2s}u'^sxv'^{2s}v^s,i} = r_{u^su'^{2s}xv'^sv^{2s},i}.\]
and hence for any $x \in \WMSigma$
\[u^su'^sxv'^sv^s \in L \Leftrightarrow u^{2s}u'^sxv'^{2s}v^s \in L \Leftrightarrow u^su'^{2s}xv'^sv^{2s} \in L.\]
Arguing that due to the choice of $s$, $u^s x v^s$ maps to $\extomega[u,v](x)$ in the syntactic $\Ext$-algebra of $L$ (and similarly for $u'^s,v'^s$, $\dots$) we obtain that
\[ \extomega[u,v](\extomega[u',v'](x)) = \extomega[u,v](\extomega[u,v'](\extomega[u',v'](x))) = \extomega[u,v](\extomega[u',v](\extomega[u',v'](x)))\]
which proves the claim.\\

We are now going to use this equation to prove for two languages, that they are not VCL. Recall that an equation $u = v$, where $u,v \in \ExtHat$, holds for a language $L$, if and only if there is a finite $\Ext$-algebra $R$ and a morphism $\psi \colon \WMSigma \rightarrow R$ recognising $L$, such that $\what{\psi}(u) = \what{\psi}(v)$. 

In particular $\what{\psi}(\extomega)$ maps to the unique idempotent generated by $\psi(\ext)$ in $\cO(R)$. As before, we assume that in the following examples, the recognising morphism is understood from the names of the elements of $\cO(R)$.

\begin{example}
Consider the language $L_\text{ML}$ \cite{Lu18} over the alphabet $\{a,b,c\}$, where $a$ is a call, $b$ a return and $c$ an internal letter, given by the production rules:
\[S \rightarrow aScb \mid acSb \mid \lambda.\]
This language is recognised by the $\Ext$-algebra
\begin{center}
\begin{tabular}[t]{c|cccccc}
$\cdot$ & 1 & $c$ & $acb$ & $acbc$ & $0$ \\ \hline
1 & 1 & $c$ & $acb$ & $acbc$ & $0$ \\ 
$c$    & $c$ & 0 & $acbc$ & $0$ & 0\\
$acb$    & $acb$ & $acbc$ & 0 & 0 & 0 \\
$acbc$    & $acbc$ & 0 & 0 & 0 & 0 \\
$0$    & $0$ & 0 & 0 & 0 & 0 \\
\end{tabular}
\hspace{1cm}
\begin{tabular}[t]{c|cccccc}
 & 1 & $c$ & $acb$ & $acbc$ & $0$ \\ \hline
$\ext$ & 0 & $acb$ & 0 & $acb$ & 0\\ 
$\ext[ac,b],\ext[a,cb]$ & $acb$ & 0 & $acb$ & 0 & 0\\ 
$\ext[ac,cb]$ & 0 & 0 & 0 & 0 & 0\\
\end{tabular}
\end{center}
Observe that each of the operations displayed on the right side is idempotent and hence for $u=ac,v=b,u'=a,v'=cb$ in Equation \ref{eq:VCL} we have
\[ \extomega[ac,b](\extomega[a,cb](1)) = \ext[ac,b](\ext[a,cb](1)) = acb \]
but 
\[\extomega[ac,b](\extomega[ac,cb](\extomega[a,cb](1))) = \ext[ac,b](\ext[ac,cb](\ext[a,cb](1))) = 0.\]
Hence $L_\text{ML}$ does not satisfy Equation \ref{eq:VCL} and thus is not VCL.
\end{example}

In the following example, it is less trivial to find the correct representatives for the operations that violate the equation.

\begin{example}
Consider the language $H^+$, which was recognised by the $\Ext$-algebra in Example \ref{ex:ExtAlgebraHPlus}. For a better intuition, we would like to recall that $H^+$ has direct connection to Boolean formulas. The main difficulty in proving that $H^+$ is not VCL via the equations in \ref{eq:VCL} is to find adequate words $u,v,u'$ and $v'$. We let $u = a^2, v= b^2, u' = a^2ba^2b$ and $v' = abb^2$. The choice is explained in the following:

Evidently $\ext[u,v] = \ext[a^2,b^2]$ is the identity, since it corresponds to double negation, and as such is idempotent. Hence $\ext[u,v] = \extomega[u,v].$
Moreover, $\ext[u',v] = \ext[a^2b,b]^2$ and $\ext[a^2b,b]$ is the constant map $1$, since multiplication by $ab$ corresponds to conjunction with \texttt{false} and $\ext[a,b]$ corresponds to negation. In particular $\ext[u',v]$ is also idempotent, which results in $\ext[u',v] = \extomega[u',v].$
By a similar argument $\ext[u,v'] = \ext[a^2,abb^2]$ is the constant map $0$. We obtain that thus independently of the choice of $x$:
\[0 = \extomega[u,v](\extomega[u',v](\extomega[u',v'](x))) \neq \extomega[u,v](\extomega[u,v'](\extomega[u',v'](x))) = 1\]
or more explicitly
\[\ext[a^2,b^2]^\omega(\ext[a^2ba^2b,b^2]^\omega(\ext[a^2ba^2b,abb^2]^\omega(x))) \neq \ext[a^2,b^2]^\omega(\ext[a^2,abb^2]^\omega(\ext[a^2ba^2b,abb^2]^\omega(x))).\]
Hence, $H^+$ does not satisfy equation \ref{eq:VCL} and thus $H^+$ is not VCL.
\end{example}

\subsection{Equations for 0-VCL}

We are now going to turn to a subclass of VCL: The visibly counter languages with threshold zero (or 0-VCL). As already used in \cite{BLS06}, they are precisely the languages $K \cap \WMSigma$, where $K$ is some regular language. For instance $\{a^nb^n \mid n \in \N\} = a^*b^* \cap \WMSigma$.

An observation, that is not immediately evident when thinking about threshold zero visibly counter languages, but becomes more plausible from an algebraic point of view is the following:

\begin{proposition}\label{prop:0VCLPseudoVar}
The 0-VCL form a pseudo-variety of VPL.
\end{proposition}

\textbf{Proof Sketch:} For each threshold zero VCL $L$, there exists a regular language $K$, such that $L = K \cap \WMSigma$. Let $M_K$ be the syntactic monoid of $K$ and $\eta_K \colon A^* \rightarrow M_K$ be its syntactic morphism. Then one may show that the set $\eta_K(\WMSigma)$ with the operations $\ext[u,v](x) = \eta_K(u) \cdot x \cdot \eta_K(v)$ where $x \in \eta_K(\WMSigma)$ and $u,v \in A^*$ such that $uv \in \WMSigma$, is an $\Ext$-algebra recognising $L$. In particular a language is 0-VCL if and only if it is recognised by an $\Ext$-algebra that is derived from a monoid in the above sense. It is not hard to show, that the class of all such $\Ext$-algebras is closed under subs, quotients and finite direct products and hence a pseudo-variety. Thus the 0-VCL are the corresponding pseudo-variety of languages.\\

Thus, by Reiterman's theorem, the 0-VCL are definable by a set of equations. The equations below together with the Equations \ref{eq:VCL} are sound, but possibly not complete.

\begin{proposition}\label{prop:EquationsZeroVCL}
Let $u,v,u',v' \in A^*$ such that $uv,uv',u'v,u'v' \in \WMSigma$ and let $L$ be a threshold zero VCL. Then $L$ satisfies the equations
\begin{equation}\label{eq:ZeroVCL}
\extomega[u,v](\extomega[u,v'](x) \cdot y \cdot \extomega[u',v](z)) = \extomega[u,v'](x) \cdot y \cdot \extomega[u',v](z).
\end{equation}
\end{proposition}

\textbf{Proof Sketch:} We show that the equations hold by algebraic means, in the following sense: As before, any 0-VCL $L$ is recognised by an $\Ext$-algebra $\eta_K(\WMSigma)$, where $\eta_K$ is the syntactic morphism of a regular language $K$ with operations $\ext[u,v](x) = \eta_K(u) \cdot x \cdot \eta_L(v)$.

Let $i \colon \WMSigma \rightarrow A^*$ be the inclusion. Then $\eta_K \circ i \colon \WMSigma \rightarrow \eta_K(\WMSigma)$ is an ($\Ext$-algebra) morphism recognising $L$. Since $i$ is uniformly continuous, it has a uniformly continuous extension $\what{i} \colon \ExtHat \rightarrow \what{A^*}$. This fact can be used to show that $L$ satisfies the equation $x = y$, where $x,y \in \ExtHat$ if and only if $\what{i}(x) = \what{i}(y)$.
One then shows that $\extomega[u,v](x)$ maps to $u^\omega \what{i}(x) v^\omega$ and hence
\[\what{i}(\extomega[u,v](\extomega[u,v'](x) \cdot y \cdot \extomega[u',v](z)))\]
is equal to
\[u^\omega u^\omega \cdot \what{i}(x) \cdot v'^\omega  \cdot \what{i}(y) \cdot  u'^\omega \cdot  \what{i}(z) \cdot v^\omega v^\omega = u^\omega  \cdot \what{i}(x) \cdot  v'^\omega \cdot  \what{i}(y) \cdot  u'^\omega \cdot  \what{i}(z) \cdot  v^\omega,\]
where the right side is equal to $\what{i}(\extomega[u,v'](x) \cdot y \cdot \extomega[u',v](z))$. Which proves the claim.\\

Observe that, for instance, the language $\{a^{n+k}b^nc^md^{m+k} \mid n,m,k \in \N\} = (a^*b^*c^*d^*) \cap \WMSigma$, where $a,c$ are call and $b,d$ return letters satisfies the equations in \ref{eq:ZeroVCL}.
\begin{example}
Consider the language $L = \{a^nb^nc^md^m \mid n,m \in \N\}$, which is recognised by the $\Ext$-algebra: 
 \begin{center}
\begin{tabular}[t]{c|ccccc}
$\cdot$ & 1 & $ab$ & $cd$ & $abcd$ \\ \hline
1    & 1 & $ab$ & $cd$ & $abcd$ \\ 
$ab$    & $ab$ & $abcd$ & 0 & 0\\
$cd$    & $cd$ & 0 & 0 & 0 \\
$abcd$    & $abcd$ & 0 & 0 & 0 \\
\end{tabular}
\hspace{1cm}
\begin{tabular}[t]{c|ccccc}
$\cdot$ & 1 & $ab$ & $cd$ & $abcd$ \\ \hline
$\ext$ & $ab$ & $ab$ & 0 & 0\\ 
$\ext[a,d], \ext[c,b]$ & 0 & 0 & 0 & 0\\ 
$\ext[c,d]$ & $cd$ & 0 & $cd$ & 0\\
$\dots$ &  &  &  & 
\end{tabular}
\end{center}
Observe that we do not display the full monoid $\cO(R)$. Here, it holds that
\[\extomega[a,d](\extomega[a,b](1) \extomega[c,d](1)) = 0 \neq abcd = \extomega[a,b](1) \extomega[c,d](1),\]
and thus $L$ does not satisfy Equation \ref{eq:ZeroVCL}, which implies that it is not $0$-VCL.
\end{example}

\section{Proofs of Section \ref{sec:application}}

\subsection{Proof of Proposition \ref{prop:EquationsVCL}}

Formally, we work with the following definition of VCA:

\begin{definition}
  A visibly counter automaton over a visibly pushdown alphabet $A$ with threshold $m$ is a tuple $(A,Q,q_0,F,\delta_0,\dots,\delta_m)$ where
\begin{itemize}
\item $A$ is a visibly pushdown alphabet
\item $Q$ is a finite set of states
\item $q_0 \in Q$ is the initial state
\item $F \subseteq Q$ is the set of final states
\item $\delta_i \colon A \times Q \rightarrow Q$ are the transition functions
\end{itemize}
\end{definition}

We define the \textit{stack-height} of a word $w \in A^*$, denoted by $\norm{w}$ inductively:
\begin{itemize}
\item If $w \in \Call$ then $\norm{w} = 1$,
\item if $w \in \Return$ then $\norm{w} = -1$,
\item if $w \in \Intern$ then $\norm{w} = 0$, and
\item if $w = w_1 \dots w_n$ with $w_i \in A$, then $\norm{w} = \sum_{i=1}^n \norm{w_i}$.
\end{itemize}

Observe that a word $w$ of length $n$ is well-matched, if $\norm{w} = 0$ and for each $i \leq n$, the condition $\norm{w_{<i}} \geq 0$ holds.

A configuration of a VCA $\cA$ is a tuple in $Q \times \N$ and we say that there exists an $a$-transition from $(q,i)$ to $(p,j)$, writing $(q,i) \xrightarrow{a} (p,j)$ if $j = i + \norm{a}$ and $p = \delta_i(a,q)$.  
The run of a VCA $\cA$ on a word $w \in \WMSigma$ of length $n$ is the sequence of tuples
\[(q_1,i_1)(q_2,i_2) \dots (q_n,i_n) \text{ with } (q_j,i_j) \xrightarrow{w_j} (q_{j+1},i_{j+1}) \text{ and } q_1 = q_0.\]
A run is accepting, if $q_n \in F$ and the language accepted by $\cA$ is the set of all words $w \in \WMSigma$ such that the run of $\cA$ on $w$ is accepting. Observe that we only consider well-matched words and it is hence not necessary to require $i_n = 0$ for an accepting run.

\begin{proof}
Recall that equation \ref{eq:VCL} is satisfied by a language $L$ if and only if for all $u,v,u',v' \in A^*$ such that $uv,$ $u'v',$ $u'v,$ $uv' \in \WMSigma$ the equality
\begin{align*}
\extomega[u,v](\extomega[u',v'](x)) &= \extomega[u,v](\extomega[u,v'](\extomega[u',v'](x))) \\
 & = \extomega[u,v](\extomega[u',v](\extomega[u',v'](x)))
\end{align*}
holds in the syntactic $\Ext$-algebra of $L$.

Assume that $L$ is recognised by a VCA $\cA$ with threshold $m$. For any $u \in A^*$ and $i \in \N$ such that $i + \norm{u} \geq 0$, we define the function
\[r_{u,i} \colon Q \rightarrow Q \text{ where } r_{u,i}(q) = p \text{ iff } (q,i) \xrightarrow{u} (p,i+\norm{u}).\]
Since $\cA$ is a VCA with threshold $m$, for any two $i,j \geq m$, we have $r_{u,i} = r_{u,j}$.
We show that if $\norm{u} > 0$, then there exists an $s \in \N$ such that $r_{u^s,i} = r_{u^{2s},i}$ for all $i \in \N$. Observe that it suffices to prove that for each $i=0,\dots,m$ there exists an $s_i$ with the property $r_{u^{s_i},i} = r_{u^{2s_i},i}$ and we obtain $s$ as their least common multiple. Hence let $i \in \N$ be arbitrary but fixed and observe that
\[r_{u^2,i} = r_{u,i+\norm{u}} \circ r_{u,i}\]
which implies that for $n \in \N$ such that $n \cdot \norm{u} < m \leq (n+1) \cdot \norm{u}$ and $l > 0$ the equality
\[r_{u^{n+l},i} = \underbrace{r_{u,m} \circ \dots \circ r_{u,m}}_{l \text{ times }} \circ r_{u,i+n\norm{u}} \circ \dots \circ r_{u,i+\norm{u}} \circ r_{u,i}\]
holds. Since the functions from $Q$ to $Q$ form a finite monoid, $r_{u,m}$ generates an idempotent such that for appropriate $l$ we obtain
\[r_{u^{n+l},i} = r_{u^{n+2l},i}.\]
Letting $s = n \cdot l$ proves that claim.
Similarly, if $\norm{u} < 0$, then there also exists an $s \in \N$ such that $r_{u^s,i} = r_{u^{2s},i}$ for all $i \geq \norm{u^{2n}}$ and also the case for $\norm{u} = 0$ follows in the same fashion.

Let $u,v,u',v' \in A^*$ such that they can be inserted in the equation above and let $s$ be their common exponent, in the sense that $r_{x^s,i} = r_{x^{2s},i}$ for $x \in \{u,u',v,v'\}$ and all $i \in \N$. This exponent again exists, since we can again choose the least common multiple of all single exponents.

 We observe that now for all $x \in \WMSigma$
\[r_{u^s x v^s,i} = r_{u^{2s} x v^{2s},i}\]
where $r_{u^s x v^s,0}(q_0)$ is a final state, if and only if the run of $\cA$ on $u^s x v^s$ is accepting and hence
\[u^s x v^s \in L \Leftrightarrow u^{2s} x v^{2s} \in L \]
Since the exponent $s$ was chosen independent of $i$ we may even derive that the two words are syntactically equivalent with respect to $L$. Hence we obtain that the syntactic image of $u^s x v^s$ is $\ext[u,v]^\omega(x)$.

Moreover, we derive that
\[r_{u^su'^s x v'^s v^s,i} = r_{u^{2s} u'^s x v'^{2s} v^s,i} = r_{u^s u'^{2s} x v'^s v^{2s}},\]
which results in the syntactic equivalence of those words and by the previous observation in the validity of equation \ref{eq:VCL}.
\end{proof}

\subsection{Proof of Proposition \ref{prop:0VCLPseudoVar}}
Denote by $\ExtMon$ the class of all $\Ext$-algebras $R$ such that there exists a finite monoid $M$ and a morphism $\phi \colon A^* \rightarrow M$, where $R$ is isomorphic to $\phi(\WMSigma)$ with the operations $\ext[u,v](x) = \phi(u) \cdot x \cdot \phi(v)$. We show that this class is a pseudo-variety.

\begin{proof}
We prove that $\ExtMon$ is closed under quotients, taking subalgebras and finite direct products.

Let $R \in \ExtMon$ and let $M$ be the monoid and $h \colon A^* \rightarrow M$ be the morphism generating $R$. Without loss of generality, we assume that $h$ is surjective, since otherwise we restrict to the image of $h$. Then $R$ is generated by the set $h(\Intern)$ and the monoid of operations $\cO(R)$ is generated by $\ext$ for $a \in \Call$ and $b \in \Return$. 

Let $S$ be a quotient of $R$. Then there exists a surjective morphism (of $\Ext$-algebras) $\varphi \colon R \rightarrow S$. 

Define the relation 
\[\sim_{\varphi} :=  \{(h(u),h(v)) \mid u,v \in \FreeExt \text{ with } \varphi(h(u)) = \varphi(h(v))\}\]
on $M$ and let $\equiv_\varphi$ be the congruence relation on $M$ generated by $\sim_\varphi$. Define the monoid $N := M\backslash_{\equiv_\varphi}$ and let $\psi \colon M \rightarrow N$ be the projection. It follows from the finiteness of $M$ that also $N$ is finite. Clearly for any $x \in M$ the equality
\[\varphi(\ext)(\varphi(x)) = \varphi(\ext(x)) = \varphi(h(a) \cdot x \cdot h(b)) = \psi(h(a)) \cdot x \cdot \psi(h(b))\]
holds and since $\varphi$ is surjective and $\bV$ closed under quotients, $S$ is in $\ExtMon$.

Let $S$ be a subalgebra of $R$. Then there exists an $n \leq |\Intern|$ and words $w_1,\dots,w_n \in \Intern^*$ and a $k \leq |\Call| = |\Return|$ and words $u_1,\dots,u_k \in \Call^*$, $v_1,\dots,v_k \in \Return^*$ with $|u_i| = |v_i|$ for $i = 1,\dots,k$ such that $S$ is generated by $h(w_1),\dots,h(w_n)$ and $\ext[u_i,v_i]$. Choose some enumeration of the call-, return- and internal letters and define the morphism of monoids $g \colon A^* \rightarrow M$ by sending the $i$th letter of $\Call$ (resp. $\Return$, $\Intern$) to $h(u_i)$ (resp. $h(v_i)$, $h(w_i)$), if $i$ does not exceed $k$ (resp. $n$) and to the neutral element of $M$ otherwise. By construction, $g$ generates $S$ and since $M$ is in $\bV$, $S$ is in $\ExtMon$.

The closure under product is clear: If $S$ and $R$ are monoidal, then $S \times R$ is generated by the product-morphism generating $S$ and $R$.
\end{proof}

By $\LExtMon$ denote the corresponding pseudo-variety. We show that a language belongs to $\LExtMon$ if and only if it is of the form $\WMSigma \cap K$ for a regular language $K$ (which is equivalent to being 0-VCL).

\begin{proof}
Suppose that $L \subseteq \FreeExt$ is in $\LExtMon$ and let $R_L$ be its syntactic $\Ext$-algebra and $\eta_L$ its syntactic morphism. Then there exists a monoid $M \in \bV$ such that $R_L$ is isomorphic to a submonoid of $M$ with $\ext(x) = x_a \cdot x \cdot x_b$ for some $x_a,x_b \in M$, where $\ext$ stands for the map $[w] \mapsto [awb]$. Define the morphism
\begin{align*}
h \colon A & \rightarrow M\\
x & \mapsto
\begin{cases}
x_a & \text{for } a \in \Call \text{ or } a \in \Return\\
\eta_L(x) & \text{otherwise.}
\end{cases}
\end{align*}
Define $S = h^\mone(h(L))$. Since $M \in \bV$ is finite, $S$ is a language in $\cV$ and it follows from the definition of $h$, that it coincides with $\eta_L$ on well-matched words, which proves one direction of the claim.

Let $S \in \cV(A)$ be a language and let $L = \FreeExt \cap S$. By $M_S$ denote the syntactic monoid of $S$ and by $\eta_S \colon A^* \rightarrow M_S$ its syntactic morphism. Then $\eta_S$ generates an $\Ext$-algebra $R \subseteq M_S$ and the morphism $\varphi \colon \FreeExt \rightarrow R$ induced by $\eta_S$ recognises $L = \varphi^\mone(\varphi(L))$. It follows that $L \in \LExtMonV$.
\end{proof}

\subsection{Proof of Proposition \ref{prop:EquationsZeroVCL}:}

Recall that $A^*$ is equipped with the metric $d(x,y) = 2^{-r(x,y)}$ where
\[r(x,y) = \{|M| \mid M \text{ is a monoid separating $x$ and $y$.}\}.\]
\begin{lemma}\label{lem:inclusionUnif}
The inclusion $i \colon \WMSigma \rightarrow A^*$ is uniformly continuous and $u = v$ holds for $\ExtMon$ if and only if $\what{i}(u) = \what{i}(v)$ in $\what{A^*}$.
\end{lemma}
\begin{proof}
Let $u,v \in \FreeExt$. We show that $d(\iota(u),\iota(v)) \leq d(u,v)$. Assume that $u$ and $v$ are separated by a monoid $M$. Then, there exists a morphism $h \colon A^* \rightarrow M$ such that $h(u) \neq h(v)$ and the $\Ext$-algebra $h(\WMSigma)$ with the multiplication of $M$ and $\ext(x) = h(a) \cdot x \cdot h(b)$ separates $u$ and $v$. Since $|h(\WMSigma)| \leq |M|$, we obtain $r(u,v) \leq r(\iota(u),\iota(v))$ and conclude $d(\iota(u),\iota(v)) \leq d(u,v)$, which proves that $\iota$ is uniformly continuous.

The direction stating that $\ExtMon$ satisfies $u = v$ implies that $\what{\iota}(u) = \what{\iota}(v)$ follows directly from the uniform continuity of $\iota$.

For the converse direction, let $M$ be a finite monoid and $R \in \ExtMon$, such that $R$ is isomorphic to some submonoid of $M$ and the operations on $R$ can be represented as usual by multiplication in $M$. Now it is clear, that if $\what{\iota}(u) \leftrightarrow \what{\iota}(v)$, then $\what{h}(\what{\iota}(u)) = \what{h}(\what{\iota}(v))$ for any morphism $h \colon A^* \rightarrow M$. Since any morphism $g \colon \FreeExt \rightarrow R$ is of the form $h \circ \iota$ for some morphism $h$, we obtain $\what{g}(u) = \what{g}(v)$, since $\what{g} = \what{h} \circ \what{\iota}$.

Observe that each morphism from $A^*$ to $M$ induces a morphism from $\FreeExt$ to $R$ and vice versa. 
\end{proof}

\begin{proof}
  It follows directly from Lemma \ref{lem:inclusionUnif}, since for $x,y,z \in \ExtHat$:
  \[\what{\iota}(\extomega[u,v](\extomega[u,v'](x)y\extomega[u',v](z))) = u^\omega u^\omega x v'^\omega y u'^\omega z v^\omega v^\omega = u^\omega x v'^\omega y u'^\omega z v^\omega. \]
  
\end{proof}

\section{Conclusion and Further Research}

We have shown that it is possible to derive a notion of equations that are capable of characterising classes of VPL. In particular, we used these results in order to give concrete equations for two subclasses of VPL: the visibly counter languages and a subclass thereof -- the visibly counter languages with threshold zero. We were only capable to show soundness of the equations for these classes, yet, where soundness means that a language in the class satisfies the equations. However, the equations found are strong enough to show for certain languages that they are not VCL or not VCL with threshold zero. We conjecture that the set of equations we gave for VCL is complete, in the sense that a VPL satisfies the equations if and only if it is VCL.

The decidability of whether a given VPL is VCL or threshold zero VCL was proven in \cite{BLS06}. If we were able to prove that the equations for threshold zero VCL are not only sound, but complete, it would imply the mentioned result of \cite{BLS06} and thus present an algebraic reproof, also improving the runtime of the currently best known decision algorithm. However, it is still unclear how hard an algebraic reproof of \cite{BLS06} might be, since the procedures used there are already very intricate.

A second direction of research is the investigation of the connection of VPL and circuit classes. Motivated by the characterisation of the regular languages in \textbf{AC}$^0$ as the quasi-aperiodic languages with the equation $(x^{\omega-1}y)^{\omega+1} = (x^{\omega-1}y)^{\omega}$ for words $x,y$ of the same length, investing in the algebraic and equational characterisation of visibly pushdown languages in \textbf{TC}$^0$ might be a step in the direction of a better understanding of this class. Here, the connection of $\Ext$-algebras to words rather than trees might prove helpful in imitating techniques such as in \cite{GPK16}, where ultrafilter equations for a certain fragment of logic gave rise to profinite equations for the regular languages therein.\\

\newpage

\bibliography{biblography.bib}

\newpage

\end{document}